\documentclass[12pt]{amsart}

\usepackage{amscd} 
\usepackage{bbold}
\usepackage{mathrsfs}
\newtheorem{theorem}{Theorem}[section]
\newtheorem{lemma}[theorem]{Lemma}
\newtheorem{corollary}[theorem]{Corollary} 
\theoremstyle{definition}

\newtheorem{remark}[theorem]{Remark} 
\numberwithin{equation}{section}

\newsymbol\subsetneq 2328
\newsymbol\dotplus 1275

\def\T{\widehat\Sigma}
\def\B{\mathscr B}

\def\uno{1}
\def\Sob{\mathfrak H}
\def\H{\mathfrak F}

\def\DS{{\dom(S^{\times})}}
\def\WD{\dom(\widehat S^{\times})}
\def \X{\mathfrak X}

\def\D{\text{\rm dom}}
\def\dom{\text{\rm dom}}
\def\ran{\text{\rm ran}}

\def\G{\mathscr G}
\def\RE{\mathbb R}
\def\CO{{\mathbb C}}

\def\ph*{\phi_\star}

\def\be{\begin{equation}}
\def\ee{\end{equation}}

\def\-{{\rm in}}
\def\+{{\rm ex}}

\def\op{\overset{\infty}{\underset{n= 0}{\oplus}}}
\def\x{\mathsf x}
\def\wR{\widehat R}

\title[On the self-adjointness of H+A${\,}^{\!\!\!*}$+A]
{On the self-adjointness of H+A${\,}^{\!\!\!*}$+A}

\author{Andrea Posilicano}
\address{DiSAT, Sezione di Matematica, Universit\`a dell'Insubria, via Valleggio 11, I-22100
Como, Italy}
\email{andrea.posilicano@unisubria.it}

\begin{document}

\begin{abstract} Let $H:\D(H)\subseteq\H\to\H$ be self-adjoint and let $A:\D(H)\to\H$ (playing the role of the annihilation operator) be $H$-bounded. Assuming some additional hypotheses on $A$ (so that the creation operator $A^{*}$ is a singular perturbation of $H$), by a twofold application of a resolvent Kre\u\i n-type formula,  we build self-adjoint realizations $\widehat H$ of the formal Hamiltonian $H+A^{*}+A$ with $\D(H)\cap\D(\widehat H)=\{0\}$.  We give an explicit characterization of $\dom(\widehat H)$ and provide a formula for the resolvent difference $(-\widehat H+z)^{-1}-(-H+z)^{-1}$. Moreover, we consider the problem of the description of $\widehat H$ as a (norm resolvent) limit of sequences of the kind $H+A^{*}_{n}+A_{n}+E_{n}$, where the $A_{n}\!$'s are regularized operators approximating $A$ and the $E_{n}$'s are suitable renormalizing bounded operators. These results show the connection between the construction of singular perturbations of self-adjoint operators by Kre\u\i n's resolvent formula and nonperturbative theory of renormalizable models in Quantum Field Theory; in particular, as an explicit example, we consider the Nelson model. 
\end{abstract}

\maketitle

\begin{section}{Introduction}
In the last few years several works  appeared where questions about the characterization of the self-adjointness domains of some renormalizable quantum fields Hamiltonians and their spectral properties were addressed (see \cite{GW1}, \cite{GW2}, \cite{GL}, \cite{LSTT}, \cite{LS}, \cite{L1}, \cite{L2}, \cite{Sch1}, \cite{Sch2}). In such papers (see also \cite{Mosh}, \cite{Thomas}, \cite{Yafaev} for some antecedent works considering simpler models) the operator theoretic framework much resembles the one 
involved in the construction of singular perturbations of self-adjoint operators (a.k.a. self-adjoint extensions of symmetric restrictions) by Kre\u\i n's type resolvent formulae (see \cite{P01}, \cite{P08} and references therein). The correspondence is exact as regards the Fermi polaron model considered in \cite{GL} (see the remark following \cite[Corollary 4.3]{GL} and our Remark \ref{Grisem}); instead, as regards the Nelson model studied in \cite{LS} (this paper was our main source of inspiration), 
the self-adjointness domain of the Nelson Hamiltonian $H_{\rm Nelson}$ there provided does not correspond, even if it has a similar structure, to the domain of a singular  perturbation of the non-interacting Hamiltonian $H_{\rm free}$. Indeed, if that were so, then, by \cite[Remark 2.10]{P01} (see also \eqref{domain} below), the domain of $H_{\rm Nelson}$ should be given by $$\{\Psi\in\H:\Psi_{0}:=\Psi+(AH_{\rm free}^{-1})^{*}\Phi\in\D(H_{\rm free}),\ A\Psi_{0}=\Theta\Phi,\ \Phi\in\D(\Theta)\}\,,$$ for some self-adjoint operator $\Theta$ (here $A$ denotes the annihilation operator) while, by \cite{LS}, $$\D(H_{\rm Nelson})=\{\Psi\in\H:\Psi+(AH_{\rm free}^{-1})^{*}\Psi\in\D(H_{\rm free})\}\,.$$ 
These two domain representations would coincide whenever $\Theta=A-A(AH_{\rm free}^{-1})^{*}$, which, beside containing the ill-defined term $A(AH_{\rm free}^{-1})^{*}$, is not even formally symmetric.  The lack of a direct correspondence between the two approaches apparently prevents the  writing of a formula for the resolvents difference $(-H_{\rm Nelson}+z)^{-1}-(-H_{\rm free}+z)^{-1}$. Such a kind of resolvent formula can help the study, beside of the spectrum, of the scattering theory (see \cite{JMPA} and references therein).
\par
Our main aim here  is to show that $H_{\rm Nelson}$ can be still obtained using the theory of singular perturbations (thus providing a resolvent formula) by applying  Kre\u\i n's formula twice: at first one singularly perturbs  $H_{\rm free}$  obtaining an intermediate Hamiltonian (related to a different physical model, see Remark \ref{Grisem}) and then one singularly perturbs the latter obtaining the Nelson Hamiltonian (such a strategy is suggested by the use of an abstract  Green-type formula, see Lemma \ref{2.1}); since for both the two operators Kre\u\i n's resolvent formula holds, by inserting the resolvent of the first operator in the resolvent formula for the second one, re-arranging and using operator block matrices, at the end one obtains a final formula for the resolvent difference $(-H_{\rm Nelson}+z)^{-1}-(-H_{\rm free}+z)^{-1}$ only containing the resolvent of $H_{\rm free}$ and the extension parameter (which is a suitable operator in Fock space), see \eqref{nelson}.
\par 
We consider also the problem of the description of $H_{\rm Nelson}$ as a (norm resolvent) limit of sequences of the kind $H_{n}:=H_{\rm free}+A^{*}_{n}+A_{n}+E_{n}$, where the $A_{n}\!$'s are the regularized annihilation operators corresponding with an ultraviolet cutoff and the $E_{n}$'s are suitable renormalizing constants. We approach this problem by employing the  resolvent formula for $H_{\rm Nelson}$ here obtained and an analogous one for the approximating $H_{n}$; this shows the role of the ever-present term of the kind $A_{n}H^{-1}_{\rm free}A_{n}^{*}$: it is due to the difference between the so-called Weyl functions (see \eqref{WF}) in the resolvents of the $H_{n}$'s and the limit one. The Weyl function of $H_{\rm Nelson}$ contains 
$A( (-AH^{-1}_{\rm free})^{*}-(A(-H_{\rm free}+{\bar z} )^{-1})^{*})$ and $(-AH^{-1}_{\rm free})^{*}$ plays the role of a regularizing term: indeed  the operator difference $ (-AH^{-1}_{\rm free})^{*}-(A(-H_{\rm free}+{\bar z} )^{-1})^{*}$  has range in the domain of $A$ while the ranges of the single terms never are. Contrarily, the Weyl function of $H_{n}$ contains $-A_{n}(-H_{\rm free}+z)^{-1}A_{n}^{*}$ only, without the need, being a bounded operator, of adding the balancing term  $-A_{n}H^{-1}_{\rm free}A_{n}^{*}$. This explain why one has to take into account such an addendum (and also a renormalizing counterterm $E_{n}$ since $A_{n}H^{-1}_{\rm free}A_{n}^{*}$ does not converge when the ultraviolet cutoff is removed) in order to approximate $H_{\rm Nelson}$ in norm resolvent sense (see Theorem \ref{bb} and Subsection \ref{QFT}). 
\par
In the present paper we embed the previous discussion in an abstract framework; thus we consider a general self-adjoint operators $H$ (playing the role of the free Hamiltonian $H_{\rm free}$) in an abstract Hilbert space $\H$ (playing the role of the Fock space) and an abstract annihilation operators $A$. 
In Section 2 we provide a self-contained presentation (with some simplifications and generalizations) of (parts of) our previous results contained in the papers  \cite{P01}, \cite{P03}, \cite{P04},  \cite{P08} that we will need later and give a results of the approximation (in norn resolvent sense) by regular perturbations of the singular perturbations here provided. In particular, in Subsection \ref{SP}, we consider the problem of the construction, by providing their resolvents, of the self-adjoint extensions of the 
symmetric restriction $S:=H|\ker(\Sigma)$, where $\Sigma:\dom(H)\to\X$ is bounded with respect to the graph norm in $\D(H)$ and $\X$ is an auxiliary Hilbert space. 
Successively, in  Section 3, we apply the previous results to the case where $\X=\H$ and $\Sigma=A$. This provides a family $H_{T}$ of self-adjoint extension of $S$, where the parameterizing operator $T$ is self-adjoint in $\H$. Then, we apply again the results in Subsection \ref{SP} now to the case where $H=H_{T}$ and $\Sigma=1-A_{*}$, $A_{*}$ a suitable left inverse of $(A(-H+{\bar z} )^{-1})^{*}$. 
The final self-adjoint operator $\widehat H_{T}$ is the one we were looking for: it can be represented as 
$\widehat H_{T}=\overline H+A^{*}+A_{T}$, where $\overline H$ is a (no more $\H$-valued) suitable closure of $H$ such that $\overline H+A^{*}$ is $\H$-valued when restricted to $\D(S^{*})$ and $A_{T}$ is an extension of the abstract annihilation operator $A$. By inserting the resolvent Kre\u\i n formula for $H_{T}$ into the one for $\widehat H_{T}$, one gets a Kre\u\i n resolvent formula for the difference $(-\widehat H_{T}+z)^{-1}-(-\widehat H+z)^{-1}$ which contains only the resolvent of $H$ and the operator $T$ (see \eqref{KF}) in Theorem \ref{resolvent}). Since $A_{T}$ has the additive representation $A_{T}=A_{0}+T$, where $A_{0}$ corresponds to the case $T=0$, $T$ enters in an additive way in the definition of $\widehat H_{T}$, i.e.,     
$\widehat H_{T}=\widehat H_{0}+T$ and so one can relax the self-adjointness request on $T$, and suppose that  $T$ is symmetric and $\widehat H_{0}$-bounded with relative bound ${\widehat a}<1$, see Theorems \ref{alo} and \ref{sb}. The same resolvent formula holds also in this case, see \eqref{KF-sb}. Notice that this does not contradict the usual parameterization of self-adjoint extensions by self-adjoint operators; indeed the true parameterizing operator turns out to be a ($T$-dependent) $2\times 2$ block operator matrix which is always self-adjoint, even in the case $T$ is merely symmetric (see Remark \ref{mer}). In Theorem \ref{bb} we address the problem of the approximation of $\widehat H_{T}$ by a sequence of regular perturbations on $H$. Finally, in Subsection \ref{QFT}, we show how, by the suitable choice $T=T_{\rm Nelson}$ provided in \cite{LS}, one obtains $\widehat H_{T_{\rm Nelson}}=H_{\rm Nelson}$, where the self-adjoint Hamiltonian $H_{\rm Nelson}$ is the one constructed in the seminal paper \cite{N}; the same kind of analysis can be applied to other renormalizable quantum field models. 
\subsection{Notations.} \begin{itemize}
\item $\D(L)$, $\ker(L)$, $\ran(L)$ denote the domain, kernel and range of the linear operator $L$ respectively; 
\item $\varrho(L)$ denotes the resolvent set of $L$;
\item $L|V$ denotes the restriction of $L$ to the subspace $V\subset\D(L)$;
\item $\B(X,Y)$ denotes the set of bounded linear operators on the Banach space $X$ to the Banach space $Y$, $\B(X):=\B(X,X)$; 
\item $\|\cdot\|_{X,Y}$ denotes the norm in $\B(X,Y)$;
\item $\|\cdot\|_{\D(L),Y}$ denotes the norm in $\B(\D(L),Y)$, where $L:\D(L)\subset X\to Y$ is  a closed linear operator and $\D(L)$ is equipped with the graph norm;
\item $\CO_{\pm}:=\{z\in\CO:\pm\text{Im}(z)>0\}$.
\end{itemize}
\vskip10pt\noindent
{\bf Acknowledgements.} The author thanks Jonas Lampart for some useful explanations, stimulating comments and bibliographic remarks.

\end{section}

\begin{section}{Singular perturbations and Kre\u\i n-type resolvent formulae.}
\begin{subsection}{Singular perturbations}\label{SP}
For convenience of the reader, in this subsection we provide a compact (almost) self-contained presentation (with some simplifications and generalizations) of parts of the results from papers \cite{P01}, \cite{P03}, \cite{P04},  \cite{P08} that we will need in the next section; we also refer to papers \cite{P04} and \cite{P08} for the comparison with other formulations (mainly with boundary triple theory, see, e.g.,  \cite[Section 7.3]{DHMdS}, \cite[Chapter 2]{BHdS}) which produce some similar outcomes. \par 
Let $$H :\dom(H)\subseteq\H\to\H$$ be a self-adjoint operator in the Hilbert space $\H$ with scalar product $\langle\cdot,\cdot\rangle$; just in order to simplify the exposition, we suppose that $ \varrho(H )\cap\RE\not=\emptyset$ (without this hypothesis some formulae become a bit longer). We introduce the following definition:
$$
\text{$\Sob_{1}$ denotes the Hilbert space given by $\D(H)$ endowed with the scalar product $\langle \cdot,\cdot\rangle_{1}$,}
$$  
$$
\langle \psi_{1},\psi_{2}\rangle_{1}:=\langle(H^{2}+1)^{1/2}\psi_{1},(H^{2}+1)^{1/2}\psi_{2}\rangle
\,;$$
$\Sob_{1}$ coincides, as a Banach space, with $\dom(H)$ equipped with the graph norm. Given a bounded linear map $$  \Sigma:{\Sob_1} \to \X\,,$$ 
$ \X$ an auxiliary Hilbert space with scalar product $(\cdot,\cdot)$, for any $z\in  \varrho(H)$ we define the linear bounded operator 
\be\label{Gz}
G_{z} : \X\to\H \,,\qquad G_{z}:=(  \Sigma R_{{\bar z} })^{*}\,,
\ee
where
$$
R_{z}:\H\to{\Sob_1}\,,\qquad R_{z}:=(-H +z )^{-1}\,.
$$
In the most typical situations, whenever $H$ is an 2nd order, elliptic differential operator, $\Sigma$ is the Dirichlet (Neumann) trace operator on the boundary of a subset of $\RE^{n}$ and $G_{z}$ is a single (double) layer operator (see, e.g.,  \cite[Example 5.5]{P08}, \cite{JDE}  and references therein).\par
We pick $\lambda_{\circ}  \in \varrho(H )\cap\RE$ and set 
\be\label{GL}
R:=R_{\lambda_{\circ} }\,,\qquad 
G:=G_{\lambda_{\circ} }\,.\ee
By first resolvent identity one has
\be\label{RG}
(z-w)R_{w}G_{z}=G_{w}-G_{z}=(z-w)R_{z}G_{w}\,.
\ee
Hence 
$$
\ran(G_{w}-G_{z})\subseteq{\Sob_1}\,,
$$
and the linear operator (playing the role of what is called a Weyl operator-valued function in boundary triple theory, see \cite{P04}, \cite[Section 7.3]{DHMdS}, \cite[Chapter 2]{BHdS})
\be\label{WF}
M_{z}:=  \Sigma(G-G_{z}): \X\to \X
\ee is well defined and bounded; by  \eqref{RG} it can be re-written as  
\be\label{Mz}
M_{z}=(z-\lambda_{\circ} )G^{*}G_{z}=(z-\lambda_{\circ} )G_{{\bar z} }^{*}G\,.
\ee
By \eqref{Mz} one gets the relations
\be\label{QF}
M_{z}^{*}=M_{{\bar z} }\,,\qquad M_{z}-M_{w}=(z-w)G^{*}_{\bar w}G_{z}\,.
\ee
Given $\Theta:\D(\Theta)\subseteq \X\to \X$ self-adjoint, we define 
\be\label{ZST}
Z_{\Sigma,\Theta}:=\{z\in \varrho(H):\text{$\Theta+M_{z}$  has inverse in $\B( \X)$}\}\,.
\ee
\begin{remark}\label{III} By $(\Theta+M_{z})^{*}=\Theta+M_{{\bar z} }$ and by \cite[Theorem 5.30, Chap. III]{K}, one has 
$$
z\in Z_{\Sigma,\Theta}\ \Rightarrow\ {\bar z} \in Z_{\Sigma,\Theta}\,.
$$
\end{remark}
\begin{theorem}\label{KFt} Let $\Sigma:\Sob_{1}\to\X$ be bounded and let $\Theta:\D(\Theta)\subseteq \X\to \X$ be self-adjoint. Suppose that  
\be\label{H1}
\text{$Z_{\Sigma,\Theta}$ is not empty}
\ee
and define 
\be\label{K}
(-H_{\Theta}+z)^{-1}:=
(-H+z)^{-1}+G_{z}(\Theta+M_{z})^{-1}G^{*}_{{\bar z} }\,,\quad z\in Z_{\Sigma,\Theta}\,.
\ee
If
\be\label{H2}
\ker(G)=\{0\}\,,\qquad \ran(G)\cap{\Sob_1}=\{0\}\,,
\ee
then \eqref{K} is the resolvent of a self-adjoint operator $H_{\Theta}$ and $Z_{\Sigma,\Theta}= \varrho(H)\cap \varrho(H_{\Theta})$; moreover
\be\label{domain}
\D(H_{\Theta})=\{\psi\in\H:\text{$\exists\,\phi\in\dom(\Theta)$ s.t. $\psi_{0}:=\psi-G\phi\in{\Sob_1}$ and $  \Sigma\psi_{0}=\Theta\phi$}\}
\ee
and 
$$
(-H_{\Theta}+\lambda_{\circ} )\psi=(-H+\lambda_{\circ} )\psi_{0}\,.
$$
\end{theorem}
\begin{proof} At first let us notice that, by $\ran(G-G_{z})\subseteq{\Sob_1}$, \eqref{H2} implies that the same relations hold for $G_{z}$ for any $z\in\varrho(H)$. By \eqref{QF}, the operator family on the righthand side of \eqref{K} (here denoted by $\breve R_{z}$) is a pseudo-resolvent (i.e., it satisfies the first resolvent identity) and $\breve R_{z}^{*}=\breve R_{{\bar z} }$ (see \cite[page 115]{P01}). Moreover, if $\psi\in \ker(\breve R_{z})$ then $(-H+z)^{-1}\psi=-G_{z}(\Theta+M_{z})^{-1}G^{*}_{{\bar z} }\psi=-G_{z}(\Theta+M_{z})^{-1}\Sigma(-H+z)^{-1}\psi$; this gives $\psi=0$ by \eqref{H2} and so $\ker(\breve R_{z})=\{0\}$. Hence, by \cite[Theorems 4.10 and 4.19]{Stone}, $\breve R_{z}$ is the resolvent of a self-adjoint operator $\breve H$ defined by 
$$
\dom(\breve H):=\ran(\breve R_{z})=\{\psi=\psi_{z}+G_{z}(\Theta+M_{z})^{-1}\Sigma \psi_{z},\ \psi_{z}\in{\Sob_1}\}\,,
$$
$$(-\breve H+z)\psi:=\breve R_{z}^{-1}\psi=(-H+z)\psi_{z}\,.$$
Let us now show that $\breve H=H_{\Theta}$.  Posing  
$\phi_z:=(\Theta+M_{z})^{-1}\Sigma \psi_{z}\in\D(\Theta)$, since the definition of $\breve H$ is $z$-independent, $\psi\in\D(\breve H)$ if and only if, for any $z\in Z_{\Sigma,\Theta}$, there exists $\psi_z\in{\Sob_1}$,  $\Sigma \psi_z=(\Theta+M_{z})\phi_{z}$, such that $\psi=\psi_z+G_z\phi_{z}$.
Then, by \eqref{RG}, 
$$
\psi_{z}-\psi_{w}=G_w\phi_{w}-G_z\phi_{z}=
G_{z}(\phi_{w}-\phi_{z})+(z-w)R_zG_w\phi_{w}\,.
$$
By \eqref{H2}, this gives $\phi_{z}=\phi_{w}$, i.e., the definition of $\phi_{z}$ is $z$-independent. 
Thus, setting $\psi_{0}:=\psi_z+(G_z-G )\phi$, one has $\psi=\psi_{0}+G \phi$, with $\psi_{0}\in{\Sob_1}$ and
$$
\Sigma \psi_{0}-\Theta\phi=\Sigma \psi_z-\Sigma(G-G_{z} )\phi-\Theta\phi=\Sigma \psi_z-(\Theta+M_{z})\phi=0\,.
$$
Therefore $\D(\breve H)\subseteq \D(H_{\Theta})$. Conversely, given $\psi=\psi_{0}+G\phi\in\D(H_{\Theta})$, defining $\psi_{z}=\psi_{0}+(G-G_{z})\phi$, one has $\psi=\psi_{z}+G_{z}\phi$ and 
$\Sigma\psi_{z}=\Sigma\psi_{0}+\Sigma(G-G_{z})\phi=(\Theta+M_{z})^{-1}\phi$, i.e. $\psi\in\D(\breve H)$; so $\D(H_{\Theta})\subseteq \D(\breve H)$ and in conclusion $\D(\breve H)=\D(H_{\Theta})$. Then, by \eqref{RG},
\begin{align*}
&(-\breve H+\lambda_{\circ} )\psi=(-H+\lambda_{\circ} )\psi_z+(\lambda_{\circ} -z)(\psi-\psi_{z})\\
=&(-H+\lambda_{\circ} )\psi_0+(-H+\lambda_{\circ} )(\psi_{z}-\psi_0)+(\lambda_{\circ} -z)G_{z}\phi\\
=&(-H+\lambda_{\circ} )\psi_0+(-H+\lambda_{\circ} )(G-G_{z})\phi-(z-\lambda_{\circ} )G_{z}\phi\\
=&(-H+\lambda_{\circ} )\psi_0\,.
\end{align*}
Finally, \cite[Theorem 2.19 and Remark 2.20]{CFP} give $Z_{\Sigma,\Theta}\not=\emptyset\Rightarrow Z_{\Sigma,\Theta}=\varrho(H)\cap\varrho(H_{\Theta})$. 
\end{proof}
\begin{remark}
Notice that, in order to prove that \eqref{K} is the resolvent of a self-adjoint operator, only the second of the two hypothesis in \eqref{H2} is required; they both provide the domain representation \eqref{domain}. In particular,  by $\psi-G\phi_{1}-(\psi-G\phi_{2})=G(\phi_{1}-\phi_{2})\in{\Sob_1}$ and by \eqref{H2},  
for any $\psi\in\D(H_{\Theta})$ there is an unique $\phi\in\H$ such that $\psi-G\phi\in{\Sob_1}$. Hence the characterization of $\D(H_{\Theta})$ in \eqref{domain} is well defined. \par
Remarks \ref{T-1}, \ref{RR} and Theorems \ref{TRK}, \ref{sb} below show that one can still have a self-adjoint operator with a resolvent given by a formula like \eqref{K} even if hypothesis \eqref{H2} does not hold true.
\end{remark}
\begin{remark}\label{obv} Obviously, if $0\in\varrho(\Theta)$ then $\lambda_{\circ} \in Z_{\Sigma,\Theta}$. In this case, whenever \eqref{H2} holds, $\lambda_{\circ} \in \varrho(H_{\Theta})$ and 
\be\label{Rl0}
(-H_{\Theta}+\lambda_{\circ} )^{-1}=(-H+\lambda_{\circ} )^{-1}+G\Theta^{-1}G^{*}\,,
\ee
\end{remark}
Regarding hypotheses \eqref{H1} and \eqref{H2}, one has the following sufficient conditions:
\begin{lemma}\label{suff} Let $\Sigma\in\B(\Sob_{1},\X)$, let $\Theta:\dom(\Theta)\subseteq\X\to\X$ be self-adjoint and let $G_{z}$ and $Z_{\Sigma,\Theta}$ be defined as in \eqref{Gz} and \eqref{ZST}. Then
$$
\text{$\ran(  \Sigma)$ dense in $ \X\quad\Leftrightarrow\quad\ker(G_{z})=\{0\}$;}
$$
$$
\text{$\ker(  \Sigma)$ dense in $\H\quad\Rightarrow\quad\ran(G_{z})\cap{\Sob_1}=\{0\}$;}
$$
$$
\text{$  \Sigma$ surjective onto $ \X\quad\Rightarrow\quad Z_{\Sigma,\Theta}\supseteq\CO\backslash\RE$.}
$$
\end{lemma}
\begin{proof} 1) By  $\ker(G_{z})=\ran(G_{{\bar z} }^{*})^{\perp}$, $\ker(G_{z})=\{0\}$ if and only if $\ran (G_{{\bar z} }^{*})=\ran(\Sigma R_{z})=\ran(  \Sigma)$ is dense. \par
2) Suppose  $G_{z}\phi=R_{z}\psi$, equivalently $(-H+z)G_{z}\phi=\psi$. Then 
$$
\langle\psi,\varphi\rangle=\langle(-H+z)G_{z}\phi,\varphi\rangle=( \phi,G^{*}_{z}(-H+{\bar z} )\varphi)=(\phi,  \Sigma\varphi)=0
$$ 
for any $\varphi\in\ker(  \Sigma)\subseteq{\Sob_1}$.  This gives $\psi=0$ whenever $\ker(  \Sigma)$ is dense in $\H$.\par
3) Let $\phi\in\D(\Theta)$, $\|\phi\|_{\X}=1$; by \eqref{QF} one gets
\be\label{Im}
\|(\Theta+M_{z})\phi\|^{2}\ge |((\Theta+M_{z})\phi,\phi)|^{2}\ge\text{Im}(z)^{2}\,\|G_{z}\phi\|^{4}\,.
\ee
Since $  \Sigma$ is surjective, $G^*_{z}=\Sigma R_{{\bar z} }$ has a closed range and so $G_{z}$ has closed range as well by the closed range theorem. Therefore, since, by point 1), $\ker(G_{z})=\{0\}$, there exists $\gamma_{\circ}>0$ such that $\|G_{z}\phi\|\ge \gamma_{\circ}\,\|\phi\|$ (see \cite[Thm. 5.2, Chap. IV]{K}). Thus, by \eqref{Im}, $\Theta+M_{z}$ has a bounded inverse and, by \cite[Thm. 5.2, Chap. IV]{K}, has a closed range. Therefore, by \eqref{Im} again, $$\dom((\Theta+M_{z})^{-1})=\ran(\Theta+M_{z})=\ker(\Theta+M_{{\bar z} })^{\perp}=\{0\}^{\perp}= \X$$ and so $(\Theta+M_{z})^{-1}\in\B( \X)$.
\end{proof}
\begin{remark} Suppose that $\ran(\Sigma)=\X$. Then, $\ran(G_{z})\cap{\Sob_1}=\{0\}$ if and only if $\ker(  \Sigma)$ is dense in $\H$ (see \cite[Lemma 2.1]{P03}).
\end{remark}
In the following by {\it symmetric operator} we mean a (not necessarily densely defined) linear operator $S:\D(S)\subseteq\H \to\H$ such that $\langle S\psi_{1},\psi_{2}\rangle=\langle \psi_{1},S\psi_{2}\rangle$ for any $\psi_{1}$ and $\psi_{2}$ belonging to $\D(S)$; whenever $S$ is densey defined, $S^{*}$ denotes its adjoint.
\begin{lemma}\label{adj}  Let $S$ be the symmetric operator $S:=H|\ker(\Sigma)$ and suppose that 
\eqref{H2} holds true; define the linear operator $$S^{\times}:\DS\subseteq\H\to\H\,,\qquad (-S^{\times}+\lambda_{\circ}  )\psi:=(-H +\lambda_{\circ}  )\psi_{0}\,,$$ 
\begin{align*}
\DS:=&\{\psi\in\H:\exists\, \phi\in\X\ \text{\rm such that}\ \psi_{0} :=\psi-G \phi\in{\Sob_1}\}\,.
\end{align*}
If $\ker(\Sigma)$ is dense in $\H$, then  $S^{\times}\subseteq S^{*}$; if furthermore $\ran(\Sigma)=\X$, then $S^{\times}=S^{*}$. If \eqref{H1} and \eqref{H2} hold then $H_{\Theta}$ is a self-adjoint extension of $S$ and $S\subseteq H_{\Theta}\subseteq S^{\times}$.  
\end{lemma}
\begin{proof} Let $\psi\in\D(S^{\times})$, $\psi=\psi_{0}+G\phi$, and $\varphi\in\dom(S)=\ker(\Sigma)$. Then, by $G^{*}=\Sigma {R}$, 
\begin{align*}
\langle\psi,(-S+\lambda_{\circ} )\varphi\rangle=&\langle\psi,(-H+\lambda_{\circ} )\varphi\rangle=\langle\psi_{0},(-H+\lambda_{\circ} )\varphi\rangle+\langle G\phi,(-H+\lambda_{\circ} )\varphi\rangle\\
=&
\langle(-H+\lambda_{\circ} )\psi_{0},\varphi\rangle+\langle \phi,G^{*}(-H+\lambda_{\circ} )\varphi\rangle=
\langle(-H+\lambda_{\circ} )\psi_{0},\varphi\rangle+\langle \phi,\Sigma\varphi\rangle\\
=&\langle(-H+\lambda_{\circ} )\psi_{0},\varphi\rangle\,.
\end{align*}
Therefore $\psi\in \D(-S^{*}+\lambda_{\circ} )=\D(S^{*})$ and $(-S^{*}+\lambda_{\circ} )\psi=(-H+\lambda_{\circ} )\psi_{0}=(-S^{\times}+\lambda_{\circ} )\psi$.  Hence $S^{\times}\subseteq S^{*}$. The equality $S^{\times}=S^{*}$ whenever $\ran(\Sigma)=\X$ is proven in \cite[Theorem 4.1]{P03}. Finally, $\ker(\Sigma)\subseteq\D(H_{\Theta})$ and $H_{\Theta}|\ker(\Sigma)=H|\ker(\Sigma)$ are immediate consequences of Theorem \ref{KFt}.
\end{proof}
\begin{lemma}\label{BT} Let $S^{\times}$ be defined as in Lemma \ref{adj}. Then, 
for any $ \psi, \varphi\in  \DS $, one has the abstract Green's identity
\be\label{Green}
\langle S^{\times} \psi, \varphi\rangle -\langle  \psi,S^{\times} \varphi\rangle =
(\Sigma_{*} \psi,  \Sigma_{0} \varphi)-(\Sigma_{0} \psi,\Sigma_{*} \varphi) \,,
\ee
where, in case $\psi\in  \DS $ decomposes as $\psi=\psi_{0}+G\phi$,
\be\label{bt1}
  \Sigma_{0} : \DS\to\X,\quad   \Sigma_{0} \psi:=  \Sigma \psi_{0} \,,
\ee
\be\label{bt2}
\Sigma_{*}: \DS\to\X,\quad \Sigma_{*}\psi:=\phi\,.
\ee
\end{lemma}
\begin{proof} Let $\psi=\psi_{0}+G\phi$, $\varphi=\varphi_{0}+G\rho$. By the definition of $S^{\times}$ and by $G^{*}=\Sigma {R}$, one gets
\begin{align*}
&\langle S^{\times} \psi, \varphi\rangle -\langle  \psi,S^{\times} \varphi\rangle =
-(\langle (-S^{\times}+\lambda_{\circ} ) \psi, \varphi\rangle -\langle  \psi,(-S^{\times}+\lambda_{\circ} ) \varphi\rangle )
\\
=&-(\langle (-H+\lambda_{\circ} ) \psi_{0}, \varphi_{0}+G\rho\rangle -\langle  \psi_{0}+G\phi,(-H+\lambda_{\circ} ) \varphi_{0}\rangle )\\
=&-(\langle  \psi_{0}, (-H+\lambda_{\circ} )\varphi_{0}\rangle+ (\Sigma\psi_{0},\rho)-\langle   \psi_{0},(-H+\lambda_{\circ} )\varphi_{0}\rangle -(\phi,\Sigma\varphi_{0}))\\
=&(\Sigma_{*}\psi,\Sigma_{0}\varphi)-(\Sigma_{0}\psi,\Sigma_{*}\varphi)\,.
\end{align*}
\end{proof}
\begin{remark} By Lemma \ref{BT}, whenever $\ker(\Sigma)$ is dense in $\H$ and $\ran(\Sigma)=\X$, the triple $(\X,\Sigma_{*},\Sigma_{0})$ is a boundary triple for $S^{*}$ (see \cite[Theorem 3.1]{P04}, \cite[Theorem 4.2]{P08}). Otherwise $(\X,\Sigma_{*},\Sigma_{0})$ resembles a boundary triple of bounded type (see \cite[Section 7.4]{DHMdS}, see also \cite[Section 6.3]{BL} for the similar definition of quasi boundary triple).
\end{remark}
\begin{remark}\label{LI} $\Sigma_{*}$ is a left inverse of $G_{z}$: since $\ran(G_{w}-G_{z})\subseteq{\Sob_1}$, one has $\Sigma_{*}G_{z}\phi=\Sigma_{*}((G_{z}-G)\phi+G\phi)=\phi$.
\end{remark}
The operator $S^{\times}$ (and hence also $H_{\Theta}$) has an alternative additive representation. 
At first, following \cite[Section 9]{KP}, we introduce a convenient scale of Hilbert spaces $\Sob_{s}$, $s\in \RE$, $\Sob_{t}\hookrightarrow\Sob_{0}\equiv\H\hookrightarrow\Sob_{u}$, $t<0<u$. We define $\Sob_{s}$ as (the completion of, whenever $s<0$) $\dom((H^{2}+1)^{s/2})$ endowed with the scalar product $$\langle \psi_{1},\psi_{2}\rangle_{s}:=\langle(H^{2}+1)^{s/2}\psi_{1},(H^{2}+1)^{s/2}\psi_{2}\rangle\,.
$$
Notice that that $R_{z}$ extends to a bounded bijective map (which we denote by the same symbol) on $\Sob_{s}$, $s<0$, and $R_{z}\in\B(\Sob_{s},\Sob_{s+1})$ for any $z\in\varrho(H)$ and for any $s\in\RE$; here we are in particular interested in the case $s=-1$. The linear operator $H$, being a densely defined bounded operator on $\H$ to  $\Sob_{-1}$, extends to the bounded operator on the whole $\H$ given by its closure: for any $\psi\in\H$ and for any sequence $\{\psi_{n}\}_{1}^{\infty}\subseteq{\Sob_1}$ such that $\psi_{n}\overset{\H}\to \psi$
$$
\overline H:\H\to \Sob_{-1}\,,\qquad \overline H\psi:=\Sob_{-1}\,\text{-}\lim_{n\uparrow\infty} H\psi_{n}\,.
$$
Let us denote by $\langle\cdot,\cdot\rangle_{-1,1}:\Sob_{-1}\times{\Sob_1}\to\CO$,  the pairing obtained by extending the scalar product:
\be\label{pairing}
\langle\psi,\varphi\rangle_{-1,1}:=\lim_{n\uparrow\infty}\langle\psi_{n},\varphi\rangle\,,\qquad \psi_{n}\overset{\Sob_{-1}}\to\psi\,,\ \psi_{n}\in\H\,,\ \varphi\in{\Sob_1}\,. 
\ee
Then we define $\Sigma^{*}:\X\to\Sob_{-1}$ by 
\be\label{menouno}
\langle \Sigma^{*}\phi,\varphi\rangle_{-1,1}=(\phi,\Sigma\varphi)\,,\qquad \varphi\in{\Sob_1}\,,\ \phi\in\X\,.
\ee
\begin{remark}\label{closed}Let us notice that $R_{z}:\Sob_{-1}\to \H$ is the adjoint, with respect the pairing $\langle\cdot,\cdot\rangle_{-1,1}$, of $R_{\bar z}:\Sob_{1}\to \H$ and it is the inverse of $(-\overline H+z):\H\to\Sob_{-1}$; therefore $G_{z}=R_{z}\Sigma^{*}$ and $$
\ker(G)=\{0\}\ \Leftrightarrow\ \ker(\Sigma^{*})=\{0\}\,,$$ 
\be\label{H2.2} \ran(G)\cap\Sob_{1}=\{0\}\ \Leftrightarrow\ \ran(\Sigma^{*})\cap\H=\{0\}\,.
\ee
If $\Sigma_{\circ}:\Sob_{1}\subseteq\H\to\X$ denotes the densely defined, linear operator  $\Sigma_{\circ}\psi:=\Sigma\psi$, then $\Sigma_{\circ}^{*}:\dom(\Sigma_{\circ}^{*})\subseteq\X\to\H$ is the restriction of $\Sigma^{*}$ to the subspace $\{\psi\in\X:\Sigma^{*}\psi\in\H\}$; therefore, by \eqref{H2.2}, $\ran(G)\cap\Sob_{1}=\{0\}$ if and only if $\dom(\Sigma_{\circ}^{*})=\ker(\Sigma^{*})$. Thus, if $\Sigma_{\circ}$ is closable, so that $\dom(\Sigma_{\circ}^{*})$ is dense, 
then the hypothesis $\ran(G)\cap\Sob_{1}=\{0\}$ is violated (here we omit the trivial case $\Sigma\equiv0$).
\end{remark}
\begin{lemma}\label{add} If $\psi\in\D(S^{\times})$ then 
$\overline H\psi+\Sigma^{*}\Sigma_{*}\psi$ belongs to $\H$ and it equals $S^{\times}\psi$:
$$
S^{\times}=(\overline H+\Sigma^{*}\Sigma_{*})|\D(S^{\times})\,.
$$
\end{lemma}
\begin{proof} Let $\psi\in\D(S^{\times})$, $\psi=\psi_{0}+G\phi$. Then 
\begin{align*}
&S^{\times}\psi=-(-S^{\times}+\lambda_{\circ} )\psi+\lambda_{\circ} \psi=-(-H+\lambda_{\circ} )\psi_{0}+\lambda_{\circ} \psi\\
=& -(-\overline H+\lambda_{\circ} )(\psi-G\phi)+\lambda_{\circ} \psi= \overline H\psi+(-\overline H+\lambda_{\circ} )G\phi\,.
\end{align*}
Noticing that, for any $\psi\in\H$ and $\varphi\in{\Sob_1}$, taking any sequence $\{\psi_{n}\}_{1}^{\infty}\subseteq{\Sob_1}$ such that $\psi_{n}\overset{\H}\to \psi$, one has 
$$
\langle(-\overline H+\lambda_{\circ} )\psi,\varphi\rangle_{-1,1}=
\lim_{n\uparrow\infty}\langle(-H+\lambda_{\circ} )\psi_{n},\varphi\rangle_{-1,1}=\lim_{n\uparrow\infty}\langle\psi_{n},(- H+\lambda_{\circ} )\varphi\rangle=\langle\psi,(- H+\lambda_{\circ} )\varphi\rangle\,,
$$ one gets
$$
\langle(-\overline H+\lambda_{\circ} )G\phi,\varphi\rangle_{-1,1}=
\langle G\phi,(-H+\lambda_{\circ} )\varphi\rangle=(\phi, G^{*}(-H+\lambda_{\circ} )\varphi)=
(\phi,\Sigma\varphi)=\langle \Sigma^{*}\phi,\varphi\rangle_{-1,1}\,.
$$
This gives $(-\overline H+\lambda_{\circ} )G\phi=\Sigma^{*}\phi=\Sigma^{*}\Sigma_{*}\psi$ and the proof is done.
\end{proof} 
Summing up, one gets the following
\begin{theorem}\label{rem} Given $\Sigma:\Sob_{1}\to\X$ bounded and $\Theta:\D(\Theta)\subseteq\X\to\X$ self-adjoint, suppose that hypotheses \eqref{H1} and \eqref{H2} hold. Then, setting 
$$
\Sigma_{\Theta}:\D(\Sigma_{\Theta})\subseteq\H\to\H\,,\qquad \Sigma_{\Theta}:=\Sigma_{0}-\Theta\Sigma_{*}\,,
$$
$$
\D(\Sigma_{\Theta}):=\{\psi\in\D(S^{\times}):\Sigma_{*}\psi\in\D(\Theta)\}\,,
$$
one has that $H_{\Theta}=S^{\times}|\ker(\Sigma_{\Theta})$ is a self-adjoint extension of $S=H|\ker(\Sigma)$; moreover
$$
H_{\Theta}=\overline H+\Sigma^{*}\Sigma_{*}
$$
and
\be\label{K1}
(-H_{\Theta}+z)^{-1}=
(-H+z)^{-1}-G_{z}(\Sigma_{\Theta}G_{z})^{-1}G^{*}_{{\bar z} }\,,\quad z\in \varrho(H)\cap
\varrho(H_{\Theta})\,.
\ee
\end{theorem}
\begin{proof} The thesis is consequence of  Theorem \ref{KFt}, Lemmata \ref{adj}, \ref{BT} and \ref{add}, noticing that, for any $\phi\in\D(\Theta)$, 
$$
(\Theta+M_{z})\phi =\Theta\phi+\Sigma(G -G_{z})\phi = -\Sigma_{0}((G_{z} -G)\phi+G\phi)+\Theta\phi =- (\Sigma_{0}-\Theta\Sigma_{*})G_{z}\phi\,.
$$  
\end{proof}
\begin{remark} Notice that if $\Theta$ has an inverse $\Lambda$ then $\Sigma_{*}\psi=\Lambda\Sigma_{0}\psi$ for any $\psi\in\D(H_{\Theta})=\ker(\Sigma_{\Theta})$; therefore in this case one has, in Theorem \ref{rem}, 
$$
H_{\Theta}=\overline H+\Sigma^{*}\Lambda\Sigma_{0}\,.
$$
\end{remark}
\end{subsection}
\begin{subsection}{Approximations by regular perturbations}
If $\Sigma$ is a bounded operator on $\H$, $\Sigma\in\B(\H,\X )$, then $G_{z}=R_{z}\Sigma^{*}$ has values in ${\Sob_1}$ and so hypothesis \eqref{H2} does not hold; more generally, by Remark \ref{closed}, hypothesis \eqref{H2} is violated whenever $\Sigma$, as an operator in $\H$ with domain $\Sob_{1}$, is closable. A simple example of an analogue of resolvent formula \eqref{K} in the case of regular perturbations is provided in the following
\begin{theorem}\label{TRK} Let $\Lambda:\dom(\Lambda)\subseteq\X\to \X$ be
symmetric and let $\Sigma_{\circ}:\dom(\Sigma_{\circ})\subseteq\H\to\X$, $\dom(\Sigma_{\circ})\supseteq\Sob_{1}$, be closable such that  
$\Sigma_{\circ}\in\B(\Sob_{1},\X)$, $\Lambda\Sigma_{\circ}\in\B(\Sob_{1},\X)$, $\Sigma_{\circ}^{*}\Lambda\Sigma_{\circ}\in\B(\Sob_{1},\H)$  and $\Lambda\Sigma_{\circ}R\Sigma_{\circ}^{*}\in\B(\X)$. If
$$
\lim_{|\gamma|\uparrow\infty}\|\Sigma_{\circ}^{*}\Lambda\Sigma_{\circ}R_{i\gamma}\|_{\H,\H}=a<1\,,\qquad \lim_{|\gamma|\uparrow\infty}\|\Lambda\Sigma_{\circ}R_{i\gamma}\Sigma_{\circ}^{*}\|_{\X,\X}=b<1\,,\qquad \gamma\in\RE\,,
$$ 
then $\widetilde H_{\Lambda}:=H+\Sigma_{\circ}^{*}\Lambda\Sigma_{\circ}$ is self-adjoint, with $\dom(\widetilde H_{\Lambda})=\Sob_{1}$ and resolvent given, whenever $|\gamma|$ is sufficiently large, by  
\be\label{K-RK}
(-\widetilde H_{\Lambda}+i\gamma)^{-1}=R_{i\gamma}+(\Sigma_{\circ}R_{-i\gamma})^{*}(\uno-\Lambda \Sigma_{\circ} R_{i\gamma}\Sigma_{\circ}^{*})^{-1}\Lambda \Sigma_{\circ} R_{i\gamma}
\,.
\ee
In the case $\Lambda=\Theta_{\circ}^{-1}$, $\Theta_{\circ}:\dom(\Theta_{\circ})\subseteq\X\to\X$ self-adjoint with $0\in \varrho(\Theta_{\circ})$, one has
\be\label{K-RK1}
(-\widetilde H_{\Lambda}+z)^{-1}=R_{z}+(\Sigma_{\circ}R_{{\bar z} })^{*}(\Theta_{\circ}- \Sigma_{\circ} R_{z}\Sigma_{\circ}^{*})^{-1} \Sigma_{\circ} R_{z}\,,\qquad z\in \varrho(H)\cap\varrho(\widetilde H_{\Lambda})\,.
\ee
\end{theorem}  
\begin{proof} At first notice that $\Lambda\Sigma_{\circ}R_{z}\Sigma_{\circ}^{*}$ is bounded for any $z\in\varrho(H)$ since both $\Lambda\Sigma_{\circ}R$ and $\Sigma_{\circ}R_{\bar z}$ are and $\Lambda\Sigma_{\circ}R_{z}\Sigma_{\circ}^{*}=\Lambda\Sigma_{\circ}R\Sigma_{\circ}^{*}+\Lambda\Sigma_{\circ}(R_{z}-R)\Sigma_{\circ}^{*}=\Lambda\Sigma_{\circ}R\Sigma_{\circ}^{*}+(\lambda_{\circ}-z)\Lambda\Sigma_{\circ}R(\Sigma_{\circ}R_{\bar z})^{*}$. 
Since $\Sigma_{\circ}$ is closable, $\Sigma_{\circ}^{*}\Lambda\Sigma_{\circ}$ is symmetric and, by our hypotheses, it is $H$-bounded with relative bound $a<1$; thus, by the Rellich-Kato theorem, $\widetilde H_{\Lambda}$ is self-adjoint with domain $\dom(\widetilde H_{\Lambda})=\Sob_{1}$. For any $\gamma\in\RE$ such that $\|\Sigma_{\circ}^{*}\Lambda\Sigma_{\circ}R_{i\gamma}\|_{\X,\X}<1$ and $\|\Lambda \Sigma_{\circ} R_{i\gamma}\Sigma_{\circ}^{*}\|_{\H,\H}<1$, one has
\begin{align*}
(-\widetilde H_{\Lambda}+i\gamma)^{-1}=&R_{i\gamma}(\uno-\Sigma_{\circ}^{*}\Lambda\Sigma_{\circ}R_{i\gamma})^{-1}
=R_{i\gamma}+\sum_{n=1}^{\infty}R_{i\gamma}(\Sigma_{\circ}^{*}\Lambda\Sigma_{\circ}R_{i\gamma})^{n}\\
=&
R_{i\gamma}+\sum_{n=1}^{\infty}\left((\Sigma_{\circ}R_{-i\gamma})^{*}(\Lambda\Sigma_{\circ}R_{i\gamma}\Sigma_{\circ}^{*})^{n-1}\,\Lambda\Sigma_{\circ}R_{i\gamma}\right)\\
=&
R_{i\gamma}+(\Sigma_{\circ}R_{-i\gamma})^{*}\Big(\sum_{n=1}^{\infty}(\Lambda\Sigma_{\circ}R_{i\gamma}\Sigma_{\circ}^{*})^{n-1}\Big)\Lambda\Sigma_{\circ}R_{i\gamma}\\
=&R_{i\gamma}+(\Sigma_{\circ}R_{-i\gamma})^{*}(\uno-\Lambda \Sigma_{\circ} R_{i\gamma}\Sigma_{\circ}^{*})^{-1}\Lambda \Sigma_{\circ} R_{i\gamma}\,.
\end{align*}
Then, by 
$(\uno-\Theta_{\circ}^{-1}\Sigma_{\circ} R_{z}\Sigma_{\circ}^{*})^{-1}\Theta_{\circ}^{-1}=(\Theta_{\circ}(\uno-\Theta_{\circ}^{-1}\Sigma_{\circ} R_{z}\Sigma_{\circ}^{*}))^{-1}=
(\Theta_{\circ}-\Sigma_{\circ} R_{z}\Sigma_{\circ}^{*})^{-1}$, one gets $
(-\widetilde H_{\Lambda}+z)^{-1}=R_{z}+(\Sigma_{\circ}R_{{\bar z} })^{*}(\Theta_{\circ}- \Sigma_{\circ} R_{z}\Sigma_{\circ}^{*})^{-1} \Sigma_{\circ} R_{z}$
for $z=i\gamma$, $|\gamma|$ sufficiently large. Finally, such a resolvent formula holds for any $z\in\varrho(H)\cap\varrho(\widetilde H_{\Lambda})$ by \cite[Theorem 2.19 and Remark 2.20]{CFP}.\end{proof}
\begin{remark}\label{RK} If $\X=\H$ and $\Sigma_{\circ}=\uno$, then Theorem \ref{TRK} is nothing else that the Rellich-Kato theorem for $H+\Lambda$. If $\X=\H$ and $V$ is self-adjoint, then, taking  $\Lambda=\text{sign}(V)$ and $\Sigma_{\circ}=|V|^{1/2}$, \eqref{K-RK} provides the Konno-Kuroda formula (due to Kato) for the resolvent of $H+V$.
\end{remark}
\begin{remark}\label{T-1} Since $\Theta_{\circ}- \Sigma_{\circ} R_{z}\Sigma_{\circ}^{*}=\Theta_{\circ}+\Sigma_{\circ}R\Sigma_{\circ}^{*}-\Sigma_{\circ}((\Sigma_{\circ}R)^{*}-(\Sigma_{\circ}R_{\bar z})^{*})$ and $\Theta_{\circ}+\Sigma_{\circ}R\Sigma_{\circ}^{*}$ is self-adjoint, \eqref{K-RK1} coincides with \eqref{K} whenever $\Sigma=\Sigma_{\circ}$ and $\Theta=\Theta_{\circ}+\Sigma_{\circ}R\Sigma_{\circ}^{*}$. However resolvent formula \eqref{K-RK1} is not a consequence of Theorem \ref{KFt}; indeed, by $\dom(\widetilde H_{\Lambda})=\Sob_{1}$ and by \eqref{K-RK1}, one has 
$\ran((\Sigma_{\circ}R_{\bar z})^{*})\cap\Sob_{1}\not=\{0\}$; this violates \eqref{H2}.
\end{remark}
In the following we use the notations $H_{\Theta}$ and $\widetilde H_{\Lambda}$ to indicate self-adjoint operators having resolvent given by formulae \eqref{K} and \eqref{K-RK} (or \eqref{K-RK1}) respectively, this independently of the validity of (some of) the hypotheses required in Theorems \ref{KFt} and \ref{TRK}.
\begin{theorem}\label{conv} Let $\Theta:\D(\Theta)\subseteq\X\to\X$ be self-adjoint, let $\Sigma\in\B({\Sob_1},\X)$ and suppose that formula \eqref{K} provides the resolvent of a  self-adjoint operator  $H_{\Theta}$. Further suppose that there exist  a sequence of closable operators $
\Sigma_{n}:\dom(\Sigma_{n})\subseteq\H\to\X$, $\dom(\Sigma_{n})\supseteq\Sob_{1}$, and a sequence of self-adjoint operators $\Theta_{n}:\dom(\Theta_{n})\subseteq\X\to\X$, $\dom(\Theta_{n})\supseteq\dom(\Theta)$, $0\in\varrho(\Theta_{n})$, such that $\Sigma_{n}\in\B(\Sob_{1},\H)$, $\Sigma_{n}R\,\Sigma_{n}^{*}\in\B(\H,\X)$ and $H+\Sigma_{n}^{*}\Lambda_{n}\Sigma_{n}$, $\Lambda_{n}:=\Theta_{n}^{-1}$, is self-adjoint with resolvent given by \eqref{K-RK1}. If  
\be\label{CS}
\lim_{n\uparrow\infty}\|\Sigma_{n}-\Sigma\|_{\Sob_{1},\X}=0\,,
\ee
\be\label{CT}
\lim_{n\uparrow\infty}\,\|(\Theta_{n}-\Sigma_{n}{R}\Sigma_{n}^{*})-\Theta\|_{\dom(\Theta),\X}=0\,,
\ee
and, in the case  of $\dom(\Theta_{n})\not=\dom(\Theta)$, 
there exist a complex conjugate couple $z_{\pm}\in\CO_{\pm}$ such that, for any  $\phi\in\X$,
\be\label{est}
\sup_{n\ge 1} \|(\Theta_{n}-\Sigma_{n}R_{z_{\pm}}\Sigma_{n}^{*})^{-1}\phi\|_{\X}<+\infty\,,
\ee
then
\be\label{don}
\lim_{n\uparrow\infty}\,(H+\Sigma_{n}^{*}\Lambda_{n}\Sigma_{n})=H_{\Theta}\quad \text{in norm-resolvent sense.}
\ee
\end{theorem} 
\begin{proof} Set $H_{n}:=H+\Sigma_{n}^{*}\Lambda_{n}\Sigma_{n}$. Given $z\in\CO\backslash\RE$, by the resolvent formulae \eqref{K} and \eqref{K-RK1} one obtains
\begin{align*}
&(-H_{n}+z)^{-1}-(H_{\Theta}+z)^{-1}=(\Sigma_{n}R_{\bar z})^{*}(\Theta_{n}-\Sigma_{n}R_{z}\Sigma_{n}^{*})^{-1}\Sigma_{n}R_{z}+G_{z}(\Sigma_{\Theta}G_{z})^{-1}G_{{\bar z} }^{*}\\
=&(\Sigma_{n}R_{\bar z})^{*}(\Theta_{n}-\Sigma_{n}R_{z}\Sigma_{n}^{*})^{-1}(\Sigma_{n}R_{z}-G^{*}_{{\bar z} })+(G_{z}-(\Sigma_{n}R_{\bar z})^{*})(\Sigma_{\Theta}G_{z})^{-1}G^{*}_{{\bar z} }\\
&+(\Sigma_{n}R_{\bar z})^{*}\big((\Theta_{n}-\Sigma_{n}R_{z}\Sigma_{n}^{*})^{-1}+(\Sigma_{\Theta}G_{z})^{-1}\big)G^{*}_{{\bar z} }\,.
\end{align*}
By the norm convergence of $(\Sigma_{n}R_{\bar z})^{*}$ and $\Sigma_{n}R_{z}$ to $G_{z}$ and $G^{*}_{{\bar z} }$ respectively, the thesis amounts to show that  
\be\label{sufff}
\lim_{n\uparrow\infty}\|(\Theta_{n}-\Sigma_{n}R_{z_{\pm} } \Sigma_{n}^{*})^{-1}+(\Sigma_{\Theta}G_{z_{\pm} } )^{-1}\|_{\H, \H}=0\,.
\ee
By hypotheses \eqref{CS}, \eqref{CT} and by the relation
\begin{align*}
&(\Theta_{n}-\Sigma_{n}R_{z}\Sigma_{n}^{*})+\Sigma_{\Theta}G_{z}\\
=&\Theta_{n}-\Sigma_{n}{R}\Sigma^{*}_{n}-\Theta+\Sigma_{n}({R}-R_{z})\Sigma^{*}_{n}
+\Sigma(G-G_{z}))\\
=&\Theta_{n}-\Sigma_{n}{R}\Sigma^{*}_{n}-\Theta+(z-\lambda_{\circ} )(\Sigma_{n}{R}(\Sigma_{n}R_{{\bar z} })^{*}-G^{*}G_{z})\,,
\end{align*}
one gets
\be\label{lTn}
\lim_{n\uparrow\infty}\,\|(\Theta_{n}-\Sigma_{n}R_{z}\Sigma_{n}^{*})+\Sigma_{\Theta}G_{z}\|_{\dom(\Theta),\X}=0\,.
\ee
Thus, by
\begin{align*}
&(\Theta_{n}-\Sigma_{n}R_{z_{\pm} }\Sigma_{n}^{*})^{-1}+(\Sigma_{\Theta}G_{z_{\pm} })^{-1}\\
=&(\Theta_{n}-\Sigma_{n}R_{z_{\pm}  }\Sigma_{n}^{*})^{-1}\big((\Theta_{n}-\Sigma_{n}R_{z_{\pm} }\Sigma_{n}^{*})+\Sigma_{\Theta}G_{z_{\pm} }\big)(\Sigma_{\Theta}G_{z_{\pm} })^{-1}\,,
\end{align*}
by the estimate
\begin{align*}
&\|(\Sigma_{\Theta}G_{z })^{-1}\|_{\H, \dom(\Theta)}=
\|(\Theta+M_{z})^{-1}\|_{\H, \dom(\Theta)}\\
\le&\|\Theta({\Theta}+M_{z})^{-1}\|_{\H, \H}+\|({\Theta}+M_{z})^{-1}\|_{\H, \H}\\
\le&\|\uno-M_{z}({\Theta}+M_{z})^{-1}\|_{\H, \H}+\|({\Theta}+M_{z})^{-1}\|_{\H, \H}<+\infty\,,
\end{align*}
and by \eqref{est} (together with uniform boundedness principle), \eqref{sufff} follows. \par 
The proof is concluded by showing that  if $\dom(\Theta_{n})=\dom(\Theta)$ then the hypothesis \eqref{est} is consequence of \eqref{CS} and \eqref{CT}. By \eqref{lTn}
and
$$
\|\Sigma_{\Theta}G_{z}\varphi\|_{\X}\ge
\|(\Sigma_{\Theta}G_{z})^{-1}\|^{-1}_{\X, \X}\|\varphi\|_{\X}\,,\qquad \varphi\in\dom(\Theta)\,,
$$
there exists $N>0$ such that, for any $n>N$ and for any  $\varphi\in\dom(\Theta)$, 
\begin{align*}
\|(\Theta_{n}-\Sigma_{n}R_{z}\Sigma_{n}^{*})\varphi\|_{\X}&\ge\|\Sigma_{\Theta}G_{z}\varphi\|_{\X}-\|(\Theta_{n}-\Sigma_{n}R_{z}\Sigma_{n}^{*})\varphi+\Sigma_{\Theta}G_{z}\varphi\|_{\X}\\
&\ge \frac12\,\|(\Sigma_{\Theta}G_{z})^{-1}\|^{-1}_{\X, \X}\|\varphi\|_{\X}\,.
\end{align*}
Therefore, choosing $\varphi=(\Theta_{n}-\Sigma_{n}R_{z}\Sigma_{n}^{*})^{-1}\phi\in\dom(\Theta_{n})=\dom(\Theta)$, 
$$
\|(\Theta_{n}-\Sigma_{n}R_{z}\Sigma_{n}^{*})^{-1}\|_{\X, \X}\le 2\,\|(\Sigma_{\Theta}G_{z})^{-1}\|_{\X, \X}\,.
$$
\end{proof}
\begin{remark}\label{Grisem} If in Theorem \ref{conv} one takes $\Theta_{n}=g_{n}^{-1}$, $g_{n}\in\RE\backslash\{0\}$ such that hypotheses there hold for some self-adjoint $\Theta$, then 
$$\lim_{n\uparrow\infty}\,(H+g_{n}\Sigma_{n}^{*}\Sigma_{n})=H_{\Theta}\quad \text{in norm-resolvent sense.}
$$
In the case $\H$ is the Fock space and $\Sigma$ is the annihilation operator (as in the next section), this (and the obvious similar version where norm-resolvent convergence is replaced by strong-resolvent convergence) is our version of \cite[Theorem 4.2]{GL}. It shows how the results provided in Subsection \ref{SP} can be used to define self-adjoint Hamiltonians describing a Fermi polaron model (see also the remark following \cite[Corollary 4.3]{GL}) and, more generally, self-adjoint operators preserving the particles number.   
\end{remark}
\end{subsection}
\end{section}

\begin{section}{Self-adjointness of $H+A^{*}+A$.}

We start by applying the results in the previous section to the case $$\X=\H\,,\qquad\Sigma=A:{\Sob_1} \to\H\,,\qquad\Theta={-T}
:\dom({T}
)\subseteq\H\to\H\,,
$$ 
where $A\in\B(\Sob_{1},\H)$ and $T$ is self-adjoint. We suppose that hypotheses \eqref{H1} and \eqref{H2} hold and so, by Theorem \ref{rem}, one gets a self-adjoint extension $H_{{T}
}$ of the symmetric operator $S=H|\ker(A)$. Here $A$ plays the role of an (abstract) annihilation operator; the change in notation is motivated by the fact that in this section we apply the previous results twice: at first with $\Sigma=A$ and then with $\Sigma$ equal to a suitable left inverse of $((-H+\bar z)^{-1}A)^{*}$.\par
  Using here the notations 
$$A_{0}\equiv\Sigma_{0}\,,\qquad A_{*}\equiv\Sigma_{*}\,,$$ 
one has (see \eqref{bt1} and \eqref{bt2}), whenever $\psi=\psi_{0}+G\phi$,
$$
A_{0} : \DS\subseteq\H\to\H\,,\quad  A_{0} \psi:=  A\psi_{0} \,,\qquad
$$
$$
A_{*}: \DS\subseteq\H\to\H\,,\quad A_{*}\psi:=\phi\,.
$$
Defining then
$$
A_{{T}
}:\D(A_{{T}
})\subseteq\H\to\H\,,\qquad A_{{T}
}:=A_{0}+{T} A_{*}\,,
$$
$$
\D(A_{{T}
}):=\{\psi\in  \DS:A_{*}\psi\in\dom({T}
)\}\,,
$$
by Theorem \ref{rem}, 
$$H_{{T}}:=S^{\times}|\ker(A_{{T}})
$$
is self-adjoint,
\be\label{RT}
(-H_{{T}
}+z)^{-1}=(-H+z)^{-1}-G_{z}(A_{{T}
} G_{z})^{-1}G_{{\bar z} }^{*}\,,\quad 
z\in \varrho(H)\cap \varrho(H_{{T}
})
\ee
and
\be\label{HT}
H_{{T}}\psi=\overline H\psi+A^{*}A_{*}\psi\,,
\ee
where $A^{*}:\H\to\Sob_{-1}$ is defined as in \eqref{menouno}.  
\par
The operator in \eqref{HT} seems to be different from what we are looking for, i.e.,  an operator of the kind $\overline H+A^{*}+A$. However, the difference is not so big: 
by the definition of $A_{{T}
}$ and by Green's formula \eqref{Green}, for any $\psi,\varphi\in \D(A_{{T}
})\subseteq\D(S^{\times})$ one has (here ${T}
$ symmetric would suffice)
\begin{align}\label{GQ}
&\langle A_{{T}
} \psi,A_{*}\varphi\rangle -\langle A_{*}\psi,A_{{T}
}\varphi\rangle\nonumber\\
=&
\langle A_{0} \psi,A_{*}\varphi\rangle -\langle A_{*}\psi,A_{0} \varphi\rangle+\langle {T}
 A_{*}\psi,A_{*}\varphi\rangle -\langle A_{*}\psi,{T}
 A_{*}\varphi\rangle\\
=&
\langle  \psi,S^{\times} \varphi\rangle-\langle S^{\times} \psi, \varphi\rangle\,.\nonumber
\end{align}
This gives the following
\begin{lemma}\label{2.1} The linear operator ${S^{\times}_{{T}
}}:\dom({S^{\times}_{{T}
}})\subseteq\H\to\H$,  ${\Sob_1}\cap\D({S^{\times}_{{T}
}})=\{0\}$, defined by 
$$
\D(S^{\times}_{{T}
}):=\{\psi\in \D(A_{{T}
}) :A_{*}\psi=\psi\}=\{\psi\in \dom({T}
):\psi-G\psi\in{\Sob_1}\}\,,
$$
\be\label{ST}
S^{\times}_{{T}
}\psi:=S^{\times}\psi+A_{{T}
}\psi\equiv
\overline H\psi+A^{*}\psi+A_{{T}
}\psi
\ee
is symmetric. 
\end{lemma}
\begin{proof} By \eqref{GQ}, for any $ \psi,\varphi\in\dom({S^{\times}_{{T}
}})$ 
one has 
$$
\langle (S^{\times}+A_{{T}
}) \psi, \varphi\rangle  =\langle  \psi,(S^{\times}+A_{{T}
}) \varphi\rangle\,,
$$
i.e., ${S^{\times}_{{T}
}}$ is symmetric. Moreover
$$
{\Sob_1}\cap\D({S^{\times}_{{T}
}})=\{\psi\in{\Sob_1}\cap\D({T}
):G\psi\in{\Sob_1}\}=\{0\}\,.
$$
\end{proof}
Since
$$
\D(H_{{T}
} )\cap\D({S^{\times}_{{T}
}})=\{\psi\in \D(H_{{T}
}):A_{*}\psi=\psi\}=
\{\psi\in \ker(A_{{T}
}):A_{*}\psi=\psi\}\,,\quad 
$$
by \eqref{HT} and \eqref{ST}, one has 
$$
{S^{\times}_{{T}
}}|\D(H_{{T}
} )\cap\D({S^{\times}_{{T}
}})=H_{{T}
}|\D(H_{{T}
} )\cap\D({S^{\times}_{{T}
}})\,,
$$
i.e.,  ${S^{\times}_{{T}
}}$ extends a restriction of a self-adjoint operator:
$${S^{\times}_{{T}
}}\supseteq\widehat  S:=H_{{T}
}|\ker(\T)\cap \D(H_{{T}
})\,,$$ where
$$
\T: \DS\to \H\,,\quad \T:=\uno-A_{*}\,.
$$
Therefore we can try to apply the formalism recalled in Subsection \ref{SP} 
to the case $H=H_{T}$  and $\Sigma=\widehat\Sigma|\dom(H_{T})$ in order to build self-adjoint extensions of $\widehat S$. If for some of such self-adjoint extensions $\widehat H$ one has $\widehat H\subseteq S^{\times}_{{T}
}$, then, since  $S^{\times}_{{T}}$ is symmetric by Lemma \ref{2.1}, $\widehat H=S^{\times}_{{T}}$ and so $S^{\times}_{{T}
}$ itself is self-adjoint. To apply such a strategy, we need to check the validity of hypotheses in Theorem \ref{KFt}. \par Since $\ker(A_{*})={\Sob_1}=\ran(R_{z})$ and $A_{*}$ is a left inverse of $G_{z}$ (see Remark \ref{LI}), 
for any $z\in Z_{\Sigma,{-T}
}$, one has 
\begin{align}\label{GB}
\T(-H_{{T}}+z)^{-1}=&(-H_{{T}}+z)^{-1}-A_{*}((-H+z)^{-1}-G_{z}(A_{{T}} G_{z})^{-1}G_{{\bar z} }^{*})\nonumber\\
=&(-H_{{T}}+z)^{-1}+(A_{{T}} G_{z})^{-1}G_{{\bar z} }^{*}
\,.
\end{align}
Thus $\T:\D(H_{{T}})\to\H$ is bounded w.r.t. the graph norm in $\D(H_{{T}})$ and, for any $z\in \varrho(H_{{T}})$ one can define the bounded operator 
$$
\widehat G_{z}:\H\to\H\,,\quad \widehat G_{z}:=\big(\T(-H_{{T}
}+{\bar z} )^{-1}\big)^{*}\,.
$$
By \eqref{GB}, for any $z\in Z_{\Sigma,{-T}
}$, one has
\be\label{wG}
\widehat G_{z}
=(-H_{{T}
}+z)^{-1}+G_{z}(A_{{T}
} G_{z})^{-1}
=(-H+z)^{-1}+G_{z}(A_{{T}
} G_{z})^{-1}(\uno-G^{*}_{{\bar z} })\,.
\ee
This shows that $$\ran(\widehat G_{z})\subseteq  \DS
$$
and $\T\widehat G_{z}$ is a well defined operator in $\B(\H)$:
\begin{align}\label{TG}
\T\widehat G_{z}=&\T(-H_{{T}}+z)^{-1}+\T G_{z}(A_{{T}} G_{z})^{-1}\nonumber\\
=&
(-H_{{T}}+z)^{-1}+(A_{{T}} G_{z})^{-1}G_{{\bar z} }^{*}+G_{z}(A_{{T}} G_{z})^{-1}-(A_{{T}} G_{z})^{-1}
\nonumber\\
=&(-H+z)^{-1}-(\uno-G_{z})(A_{{T}} G_{z})^{-1}(\uno-G^{*}_{{\bar z} })\,.
\end{align}
Regarding the validity of hypothesis  \eqref{H2}, one has the following: 
\begin{lemma}\label{cap} For any $z\in \varrho(H)\cap\varrho(H_{{T}
})$, one has
$$\ker(\widehat G_{z})=\{0\}=\ran(\widehat G_{z})\cap\D(H_{{T}
})\,.
$$
\end{lemma}
\begin{proof} At first notice that, since $A_{{T}
}(-H_{{T}
}+z)^{-1}=0$, $ A_{{T}
}\widehat G_{z}=\uno$ by \eqref{wG}. Hence $\widehat G_{z}\phi=0$ implies $0=A_{{T}
}\widehat G_{z}\phi=\phi$. Now suppose that $\widehat G_{z}\phi\in\dom(H_{{T}
})=\ker(A_{{T}
})$. Then $0=A_{{T}
}\widehat G_{z}\phi=\phi$ and so $\widehat G_{z}\phi=0$.
\end{proof}
Now, let us suppose that $\RE\cap  \varrho(H)\cap \varrho(H_{{T}
})$ is not empty (this hypothesis is not necessary, it is used in order to simplify the exposition), pick ${\widehat\lambda_{\circ} }$ there and set $$\widehat G:=\widehat G_{{\widehat\lambda_{\circ} }}\,.$$ 
By Remark \ref{obv}, one can take $\widehat\lambda_{\circ}=\lambda_{\circ}$ whenever $0\in\varrho(T)$.
\par
Define, as in Lemma \ref{adj}, $\widehat S^{\times}:\WD\subseteq\H\to\H$ by
\begin{align*}
\WD :=&\{\psi\in\H:\exists\, \phi\in\H\ \text{\rm such that}\ \widehat\psi_{0} :=\psi-\widehat G \phi\in\D(H_{{T}
} )\}\,,
\end{align*}
$$
(-\widehat S^{\times}+{\widehat\lambda_{\circ} }\, )\psi:=(-H_{{T}
} +{\widehat\lambda_{\circ} }\, )\widehat\psi_{0} \,.
$$
Then
\begin{lemma}\label{ss} One has $\WD\subseteq\DS$ and
$$
\widehat S^{\times}| \WD\cap\ker(\T)\subseteq{S^{\times}_{{T}
}}\,.
$$
\end{lemma}
\begin{proof} At first notice that, for any $\psi\in\WD$ decomposed as $\psi=\widehat\psi_{0}+\widehat G\phi$, where $\widehat\psi_{0}\in\D(H_{{T}
})$ and $\phi\in\H$, one has, since $\dom(H_{{T}
})=\ker(A_{{T}
})$ and $A_{{T}
}\widehat G=\uno$ (see the proof of Lemma \ref{cap}), 
\be\label{AT}
A_{{T}
}\psi=A_{{T}
}\widehat\psi_{0}+A_{{T}
}\widehat G\phi=\phi\,.
\ee
Since, by \eqref{wG},
$$
\psi=\widehat\psi_{0} +\widehat G \phi=\widehat\psi_{0}+(-H+{\widehat\lambda_{\circ} })^{-1}\phi+
G_{{\widehat\lambda_{\circ} }}(A_{{T}
} G_{{\widehat\lambda_{\circ} }})^{-1}(\uno-G^{*}_{{\widehat\lambda_{\circ} }})\phi
$$
and since $\ran((A_{{T}
} G_{{\widehat\lambda_{\circ} }})^{-1})=\D({T}
)$, one gets
$$
 \WD\subseteq\{\psi\in \DS:A_{*}\psi\in\dom({T}
)\}\subseteq\DS\,.
$$
By $H_{{T}
}\subseteq S^{\times}$, by $(-S^{\times} +{\widehat\lambda_{\circ} }\, )(-H+{\widehat\lambda_{\circ} }\,)^{-1}=\uno$, by $\ran(G_{{\widehat\lambda_{\circ} }})=\ker(-S^{\times}+{\widehat\lambda_{\circ} }\,)$, by \eqref{wG} and by \eqref{AT}, then one gets 
\begin{align*}
\widehat S^{\times}\psi=&-(-H_{{T}
}+{\widehat\lambda_{\circ} }\, )\widehat\psi_{0}+{\widehat\lambda_{\circ} }\psi=-(-S^{\times}+{\widehat\lambda_{\circ} }\, )\widehat\psi_{0}+{\widehat\lambda_{\circ} }\psi\\
=&-(-S^{\times} +{\widehat\lambda_{\circ} }\, )(\psi-\widehat G\phi)+{\widehat\lambda_{\circ} }\psi
=
S^{\times}\psi+(-S^{\times} +{\widehat\lambda_{\circ} }\, )\widehat G\phi\\
=&S^{\times}\psi+\phi=(S^{\times}+A_{{T}
})\psi\,.
\end{align*}
Hence, since 
$$
 \WD\cap\ker(\T)\subseteq\{\psi\in\dom({T}
):\psi-G\psi\in{\Sob_1}\}=\dom(S^{\times}_{{T}
})\,,
$$ 
the proof is done.
\end{proof}
Putting together the previous results, one gets the following 
\begin{theorem}\label{resolvent} Let ${T}
:\D({T}
)\subseteq\H\to\H$ be self-adjoint and $A:\Sob_{1}\to \H$ be bounded such that hypotheses \eqref{H1} and \eqref{H2} hold true.
If there exists 
$z_{\circ}\in \varrho(H_{{T}
})$ such that  $\T\widehat G_{z_{\circ}}$ has a bounded inverse, then  $\widehat H_{{T}
}={S^{\times}_{{T}
}}$ is self-adjoint, $\D(H)\cap\D(\widehat H_{{T}
})=\{0\}$ and
\be\label{domdom}
\D(\widehat H_{{T}
})=\{\psi\in\D({T}
):\psi-G\psi\in{\Sob_1}\}\,,
\ee
\be\label{whh}
\widehat H_{{T}
}=\overline{H}+A^{*}+A_{{T}}\,.
\ee
Moreover  
$\T\widehat G_{z}$ has a bounded inverse for any $z\in  \varrho(H_{{T}
})\cap \varrho (\widehat H_{{T}})$ and 
\begin{align*}
&
(-\widehat H_{{T}
}+z)^{-1}
=(-H_{{T}
}+z)^{-1}-\widehat G_{z}(\T\widehat G_{z})^{-1}\widehat G^{*}_{{\bar z} }\,;
\end{align*}
if  $z\in\varrho(H)\cap \varrho(H_{{T}
})\cap \varrho (\widehat H_{{T}})$ then 
\begin{align}\label{KF}
(-\widehat H_{{T}
}+z)^{-1}
=&(-H+z)^{-1}-
\begin{bmatrix}G_{z}&R_{z}\end{bmatrix}\begin{bmatrix}A_{{T}
} G_{z}&G^{*}_{{\bar z} }-\uno\\G_{z}-\uno&R_{z}
\end{bmatrix}^{\!-1}\begin{bmatrix}G_{{\bar z} }^{*}\\R_{z}\end{bmatrix}\,.
\end{align}
\end{theorem}
\begin{proof} 
By Lemma \ref{ss}, since ${S^{\times}_{{T}
}}$ is symmetric, if $\widehat H_{{T}
}:=\widehat S^{\times}| \WD\cap\ker(\T)$ is self-adjoint then $\widehat H_{{T}
}={S^{\times}_{{T}
}}$. 
This holds by Lemma \ref{cap}, by Theorem \ref{KFt}  and Theorem \ref{rem} applied to the case 
$$H=H_{{T}}\,,\qquad \Sigma=\T|\dom(H_{{T}})\,,\qquad\Theta=-\T\widehat G\,.$$ 
Notice that, by these choices, $$\Sigma_{\Theta}\psi=\widehat\Sigma\widehat\psi_{0}+\widehat\Sigma\widehat G\phi=\widehat\Sigma\psi$$ and so \eqref{domdom} holds. By \eqref{H2}, $G\psi\in\Sob_{1}$ if and only if $\psi=0$; hence 
$$\{\psi\in\D({T}):\psi-G\psi\in{\Sob_1}\}\cap\Sob_{1}=\{0\}\,.
$$ Then to conclude we only need to prove \eqref{KF}. By \eqref{K1}, \eqref{RT}, \eqref{wG} and \eqref{TG}, one gets
\begin{align*} 
&(-\widehat H_{{T}
}+z)^{-1}=
(-H_{{T}
}+z)^{-1}-\widehat G_{z}(\T\widehat G_{z})^{-1}\widehat G^{*}_{{\bar z} }
=(-H+z)^{-1}-G_{z}(A_{{T}
} G_{z})^{-1}G_{{\bar z} }^{*}\\
&-\big((-H+z)^{-1}+G_{z}(A_{{T}
} G_{z})^{-1}(\uno-G^{*}_{{\bar z} })\big)
(\T\widehat G_{z})^{-1}
\big((-H+z)^{-1}+(\uno-G_{z})(A_{{T}
} G_{z})^{-1}G^{*}_{{\bar z} }\big)
\\
=&(-H+z)^{-1}-\begin{bmatrix}G_{z}&R_{z}\end{bmatrix}{\mathbb M}
\begin{bmatrix}G_{{\bar z} }^{*}\\R_{z}\end{bmatrix}\,,
\end{align*}
where ${\mathbb M}$ is the block operator matrix 
$
{\mathbb M}=\begin{bmatrix}{M}_{11}&{M}_{12}\\
{M}_{21}&{M}_{22}\end{bmatrix}$
with entries
\begin{align*}
&{M}_{11}=(A_{{T}
} G_{z})^{-1}+(A_{{T}
} G_{z})^{-1}(\uno-G^{*}_{{\bar z} })(\T\widehat G_{z})^{-1}(\uno-G_{z})(A_{{T}
} G_{z})^{-1}\\
=&(A_{{T}
} G_{z})^{-1}+(A_{{T}
} G_{z})^{-1}(\uno-G^{*}_{{\bar z} })
\left((-H+z)^{-1}-(\uno-G_{z})(A_{{T}
} G_{z})^{-1}(\uno-G^{*}_{{\bar z} })\right)^{-1}\times\\
&\times(\uno-G_{z})(A_{{T}
} G_{z})^{-1}\,,
\end{align*}
\begin{align*}
&{M}_{12}=(A_{{T}
} G_{z})^{-1}(\uno-G^{*}_{{\bar z} })(\T\widehat G_{z})^{-1}\\
=&(A_{{T}
} G_{z})^{-1}(\uno-G^{*}_{{\bar z} })\left((-H+z)^{-1}-(\uno-G_{z})(A_{{T}
} G_{z})^{-1}(\uno-G^{*}_{{\bar z} })\right)^{-1}\,,
\end{align*}
\begin{align*}
&{M}_{21}=
(\T\widehat G_{z})^{-1}(\uno-G_{z})(A_{{T}
} G_{z})^{-1}\\
=&
\left((-H+z)^{-1}-(\uno-G_{z})(A_{{T}
} G_{z})^{-1}(\uno-G^{*}_{{\bar z} })\right)^{-1}(\uno-G_{z})(A_{{T}
} G_{z})^{-1}\,,
\end{align*}
$${M}_{22}=(\T\widehat G_{z})^{-1}=\left((-H+z)^{-1}-(\uno-G_{z})(A_{{T}
} G_{z})^{-1}(\uno-G^{*}_{{\bar z} })\right)^{-1}\,.
$$
Finally, one checks that 
$$
{\mathbb M}\begin{bmatrix}A_{{T}
} G_{z}&G^{*}_{{\bar z} }-\uno\\G_{z}-\uno&R_{z}
\end{bmatrix}=\begin{bmatrix}A_{{T}
} G_{z}&G^{*}_{{\bar z} }-\uno\\G_{z}-\uno&R_{z}
\end{bmatrix}{\mathbb M}={\mathbb 1}=\begin{bmatrix}\uno&0\\0&\uno
\end{bmatrix}\,,
$$
i.e., 
$$
{\mathbb M}=\begin{bmatrix}A_{{T}
} G_{z}&G^{*}_{{\bar z} }-\uno\\G_{z}-\uno&R_{z}
\end{bmatrix}^{-1}
$$
and the proof is done.
\end{proof}
In the next remark and below, we use the notations introduced in the previous section with letters in blackboard bold style to denote block matrix operators.
\begin{remark}\label{RR} Let the hypotheses in Theorem \ref{resolvent} hold. Noticing that
$$
\begin{bmatrix}A_{{T}
} G_{z}&G^{*}_{{\bar z} }-\uno\\G_{z}-\uno&R_{z}
\end{bmatrix}=-(\mathbb{\Theta
}_{T}+\mathbb{\Sigma}(\mathbb{G}-\mathbb{G}_{z}))\equiv\mathbb{\Sigma}_{\mathbb{\Theta
}_{T}}\mathbb{G}_{z}\,,
$$
where
$$
\mathbb{\Sigma}:{\Sob_1}\to\H\oplus\H\,,\qquad \mathbb{\Sigma}\psi:=A\psi\oplus\psi\,,
$$
$$
\mathbb{G}_{z}:\H\oplus\H\to\H\,,\qquad \mathbb{G}_{z}:=(\mathbb{\Sigma}R_{{\bar z} })^{*}\,,\quad \mathbb{G}:=\mathbb{G}_{\lambda_{\circ} }\,,
$$
and 
$$
\mathbb{\Theta
}_{T}:\D({T}
)\oplus\H\subseteq\H\oplus\H\to \H\oplus\H\,,\quad \mathbb{\Theta
}_{T}:=\begin{bmatrix}-{T}
 &\uno-G^{*}\\\uno-G&-{R}
\end{bmatrix}\,,
$$
one gets 
$$
\widehat H_{{T}}=H_{\mathbb{\Theta}_{T}}
$$
and \eqref{KF} is rewritten as (compare with \eqref{K}) 
\be\label{Kbb}
(-\widehat H_{{T}}+z)^{-1}\equiv(-H_{\mathbb{\Theta
}_{T}}+z)^{-1}=(-H+z)^{-1}-\mathbb{G}_{z}(\mathbb{\Sigma}_{\mathbb{\Theta
}_{T}}\mathbb{G}_{z})^{-1}\mathbb{G}_{{\bar z} }^{*}\,.
\ee
Since $\mathbb{G}(\psi_{1}\oplus\psi_{2})=G\psi_{1}+R\psi_{2}$, one has $\ran(\mathbb{G})\cap\Sob_{1}=\Sob_{1}$; this shows that \eqref{K} can still represent the resolvent of a self-adjoint operator even if hypotheses \eqref{H2} in Theorem \ref{KFt} does not hold. 
\end{remark}
In order to apply Theorem \ref{resolvent} one needs to show that there exists at least one $z_{\circ}\in\varrho(H)$ such that $\T\widehat G_{z_{\circ}}$ has a bounded inverse. A simple criterion is provided in the next Lemma. 
\begin{lemma}\label{sl} Let $A\in\B(\Sob_{s},\H)$ for some $s\in(0,1)$. Then
$$(\uno-G_{\pm i\gamma})^{-1}\in\B(\H)\,,\qquad(\uno-G_{\pm i\gamma}^{*})^{-1}\in\B(\H)
$$ 
whenever $\gamma\in\RE$ and $|\gamma|$ is  sufficiently large. \par 
Further suppose that $T:\D(T)\subseteq\H\to\H$ is self-adjoint such that $Z_{A,-T}\not=\emptyset$ and  $\D(T)\supseteq\Sob_{t}$, $T|\Sob_{t}\in\B(\Sob_{t},\H)$ for some $t\in [0,1-s)$. Then  
$$
(\T\widehat G_{\pm i\gamma})^{-1}\in\B(\H)
$$
whenever $\gamma\in\RE$ and $|\gamma|$ is  sufficiently large.
\end{lemma}
\begin{proof} Let us take $|\gamma|\ge 1$. By  
$$
\|(-H\pm i\gamma)^{-1}\|_{\H, \H}\le \frac1{|\gamma|}\,,\qquad \|(-H\pm i\gamma)^{-1}\|_{\H,\Sob_{1}}\le 1\,,
$$
since interpolating theorems hold for Hilbert scales of the kind $\Sob_{s}$, $s\in\RE$, (see \cite[Section 9]{KP}), one gets
$$
\|(-H\pm i\gamma)^{-1}\|_{\H,\Sob_{r}}\le \frac1{|\gamma|^{1-r}}\,,\quad 0\le r\le 1\,.
$$
Thus,  
\be\label{t-u}
\|(-H\pm i\gamma )^{-1}\|_{\Sob_{u},\Sob_{t}}=\|(-H\pm i\gamma )^{-1}\|_{\H,\Sob_{t-u}}\le \frac1{|\gamma| ^{1-(t-u)}}\,,\quad 0\le t-u\le 1\,.
\ee
Hence
$$\|G_{\mp i\gamma}^{*}\|_{\Sob_{t}, \H}\le\|A\|_{\Sob_{s},\H}\|(-H\pm i\gamma )^{-1}\|_{\Sob_{t}, \Sob_{s}}\le \frac{
\|A\|_{\Sob_{s},\H}}{|\gamma| ^{1-(s-t)}}
$$
and
$$\|G_{\pm i\gamma}\|_{\Sob_{t}, \Sob_{t}}\le
\|(-H\pm i\gamma )^{-1}\|_{\Sob_{-s}, \Sob_{t}}\|A^{*}\|_{\Sob_{t}, \Sob_{-s}}
=\|(-H\pm i\gamma )^{-1}\|_{\Sob_{-s}, \Sob_{t}}\|A\|_{\Sob_{s}, \Sob_{-t}}\le
\frac{\|A\|_{\Sob_{s}, \H}}{|\gamma| ^{1-(t+s)}}
\,.
$$
This shows that both $\uno-G_{\pm i\gamma}:\Sob_{t}\to \Sob_{t}$ and $\uno-G_{\mp i\gamma} ^{*}:\H\to\H$ have bounded inverses whenever $|\gamma| $ is sufficiently large. \par 
Since $Z_{A,-T}\not=\emptyset$, by \cite[Theorem 2.19 and Remark 2.20]{CFP}, $A_{T}G_{z}$ has a bounded inverse for any $z\in\varrho(H)\cap\varrho(H_{T})\subseteq\CO\backslash\RE$ and so
\begin{align*}
&\T\widehat G_{\pm i\gamma}\\
=&(\uno-G_{\pm i\gamma})(A_{{T}
} G_{\pm i\gamma})^{-1}\big(A_{{T}} G_{\pm i\gamma }(\uno-G_{\pm i\gamma })^{-1}(-H\pm i\gamma )^{-1}(\uno-G_{\mp i\gamma} ^{*})^{-1}-\uno\big)(\uno-G_{\mp i\gamma} ^{*})
\,.
\end{align*}
Since
$$
\|(\uno-G^{*}_{\pm i\gamma})^{-1}\|_{\H, \H}\le \sum_{n=0}^{\infty}\|G^{*}_{\pm i\gamma}\|^{n}_{\H, \H}=\frac{1}{1-|\gamma| ^{s-1}\|A\|_{\Sob_{s},\H}}\le c_{0}\,,
$$
$$
\|(\uno-G_{\mp i\gamma} )^{-1}\|_{\Sob_{t}, \Sob_{t}}\le \sum_{n=0}^{\infty}\|G_{\mp i\gamma} \|^{n}_{\Sob_{t}, \Sob_{t}}=\frac{1}{1-|\gamma| ^{t+s-1}\|A\|_{\Sob_{s},\H}}\le c_{t}
$$
and
\begin{align*}
&\|A_{{T}
} G_{\pm i\gamma}\|_{\Sob_{t}, \H}\le \|{T}
\|_{\Sob_{t}, \H}+\|M_{\pm i\gamma }\|_{\Sob_{t}, \H}\\
\le& 
\|{T}
\|_{\Sob_{t}, \H}+|\pm i\gamma-\lambda_{\circ} |\,\|G^{*}\|_{\Sob_{t}, \H}\|G_{\pm i\gamma }\|_{\Sob_{t},\Sob_{t}}\\
\le&
\|{T}
\|_{\Sob_{t}, \H}+\frac{|\lambda_{\circ} |+|\gamma| }{|\gamma| ^{2(1-s)}}\ \|A\|_{\Sob_{s},\H}^{2}\le \kappa_{t,s}\left(1+\frac{|\lambda_{\circ} |+|\gamma| }{|\gamma| ^{2(1-s)}}\,\right)
\end{align*}
one has 
\begin{align*}
&\|A_{{T}} G_{\pm i\gamma}(\uno-G_{\pm i\gamma})^{-1}(-H\pm i\gamma )^{-1}(\uno-G_{\mp i\gamma} ^{*})^{-1}\|_{\H, \H}\\
\le& 
\|A_{{T}} G_{\pm i\gamma }\|_{\Sob_{t}, \H}\|(\uno-G_{\pm i\gamma })^{-1}\|_{\Sob_{t},\Sob_{t}}
\|(-H\pm i\gamma )^{-1}\|_{\H, \Sob_{t}}\|(\uno-G^{*}_{\pm i\gamma })^{-1}\|_{\H, \H}\\
\le &\kappa_{t,s}c_{0}c_{t}\left(1+\frac{|\lambda_{\circ} |+|\gamma| }{|\gamma| ^{2(1-s)}}\,\right)\frac{1}{|\gamma| ^{1-t}}<1
\end{align*}
whenever $|\gamma| $ is sufficiently large. Hence, whenever  $|\gamma| $ is sufficiently large, $\T\widehat G_{\pm i\gamma}$ has a bounded inverse given by 
\begin{align*}
&(\T\widehat G_{\pm i\gamma})^{-1}
\\
=&(\uno-G_{\mp i\gamma} ^{*})^{-1}\big(A_{{T}} G_{\pm i\gamma }(\uno-G_{\pm i\gamma})^{-1}(-H\pm i\gamma )^{-1}(\uno-G_{\mp i\gamma} ^{*})^{-1}-\uno\big)^{-1}A_{{T}
} G_{\pm i\gamma }(\uno-G_{\pm i\gamma})^{-1}\,.
\end{align*}
\end{proof}
\begin{corollary}\label{c-sl} Let $A$ and $T$ satisfy the hypotheses in Lemma \ref{sl} and further suppose that 
both $\ker(A|\Sob_{1})$ and $\ran(A|\Sob_{1})$ are dense in $\H$. Then $\widehat H_{T}=\overline H+A^{*}+A_{T}$ is self-adjoint with domain $\D(\widehat H_{T})=\{\psi\in\Sob_{1-s}:\psi-G\psi\in\Sob_{1}\}$ and resolvent given by formula \eqref{KF}. 
\end{corollary}
\begin{proof} Theorem \ref{resolvent} and Lemmata \ref{sl} and \ref{suff} give the thesis with 
$\D(\widehat H_{T})=\{\psi\in\Sob_{t}:\psi-G\psi\in\Sob_{1}\}$. Then, by using the $\Sob_{-s}$-$\Sob_{s}$ duality, $G=RA^{*}$, where $A^{*}\in\B(\H,\Sob_{-s})$ and $R\in\B(\Sob_{s},\Sob_{1-s})$; therefore $G\in\B(\H,\Sob_{1-s})$ and the proof is concluded noticing that $\psi\in\D(\widehat H_{T})$ belongs to $\Sob_{1-s}$ if and only if $G\psi\in \Sob_{1-s}$. 
\end{proof}
\begin{remark}\label{prev}
As the proof of previous Lemma \ref{sl} shows, if $H$ is semibounded then the same conclusions there hold with $\pm i\gamma$ replaced by $\lambda\in\RE$ sufficiently far away from $\sigma(H)$.   
\end{remark}
Since the operator ${T}
$ enters as an additive perturbation in the definition of $\widehat H_{{T}
}$, one can eventually avoid the self-adjointness hypothesis on it and work with $\widehat H_{0}$ alone:
\begin{theorem}\label{alo} Let $A\in\B(\Sob_{s},\H)$ for some $0<s<1$ and such that both $\ker(A|\Sob_{1})$ and $\ran(A|\Sob_{1})$ are dense in $\H$.  Then $\widehat H_{0}:=\overline{H}+A^{*}+A_{0}$ is self-adjoint with 
domain 
$$\D(\widehat H_{0})=\{\psi\in\Sob_{1-s}:\psi-G\psi\in\Sob_{1}\}$$
and resolvent given, for any $z\in\CO$ such that $\mu+z\in  \varrho(H)\cap \varrho (\widehat H_{0})$, $\mu\in\RE\backslash\{0\}$, by 
\be\label{res-c}
(-\widehat H_{0}+z)^{-1}
=(-H+\mu+z)^{-1}-
\begin{bmatrix}G_{\mu+z}&R_{\mu+z}\end{bmatrix}\begin{bmatrix}A_{\mu} G_{\mu+z}&G^{*}_{\mu+{\bar z} }-\uno\\G_{\mu+z}-\uno&R_{\mu+z}
\end{bmatrix}^{-1}\begin{bmatrix}G_{\mu+{\bar z} }^{*}\\R_{\mu+z}\end{bmatrix}\,.
\ee
If ${T}:\D(T)\subseteq\H\to\H$, $\D(T)\supseteq\D(\widehat H_{0})$, is symmetric and $\widehat H_{0}$-bounded with relative bound ${\widehat a}<1$ then 
$\widehat H_{{T}
}:=\overline{H}+A^{*}+A_{{T}}$ is self-adjoint, has domain 
$\D(\widehat H_{{T}
})=\D(\widehat H_{0})$ 
and resolvent
\be\label{res-T}
(-\widehat H_{{T}
}+z)^{-1}=(-\widehat H_{0}+z)^{-1}+(-\widehat H_{0}+z)^{-1}(1-{T}
(-\widehat H_{0}+z)^{-1})^{-1}{T}
(-\widehat H_{0}+z)^{-1}\,.
\ee
\end{theorem}
\begin{proof} By Remark \ref{obv} and Lemma \ref{suff}, hypotheses \eqref{H1} and \eqref{H2} are satisfied with $\Theta=-T
=-\mu\not=0$. Hence, by Lemma \ref{sl} and Theorem \ref{resolvent}, 
$\widehat H_{\mu}$ (i.e., $\widehat H_{T}$ with $T=\mu$) is selfadjoint with domain $\D(\widehat H_{\mu})=\{\psi\in\H:\psi-G\psi\in{\Sob_1}\}$ and resolvent $(-\widehat H_{\mu}+z)^{-1}=(-H_{\mu}+z)^{-1}-\widehat G_{z}
(\T\widehat G_{z})^{-1}\widehat G^{*}_{{\bar z} }$. Therefore $\widehat H_{0}=\widehat H_{\mu}-\mu$ is self-adjoint with domain $\D(\widehat H_{0})=\D(\widehat H_{\mu})$ and resolvent $(-\widehat H_{0}+z)^{-1}=(-\widehat H_{\mu}+\mu+z)^{-1}$. Since $A\in\B(\Sob_{s},\H)$, one gets $\D(\widehat H_{0})\subseteq\Sob_{1-s}$ by the same arguments as in the proof of Corollary \ref{c-sl}.
Formula \eqref{res-T} is consequence of  $\widehat H_{{T}}=\widehat H_{0}+{T}
$ and Remark \ref{RK}.
\end{proof}
The next result shows how to obtain $\widehat H_{{T}}$ as limits of  regular perturbations of $H$.
\begin{theorem}\label{bb} Suppose that the operator $$\widehat H_{0}:=\overline{H}+A^{*}+A_{0}\,,\quad\D(\widehat H_{0})=\{\psi\in\H:\psi-G\psi\in{\Sob_1}\}$$ is self-adjoint with resolvent given by \eqref{res-c} for some $\mu\in\RE$. Let $A_{n}:\dom(A_{n})\subseteq\H\to \H$
be a sequence of closable operators such that, for some $s\in\left[0,\frac12\right]$,  
$$\dom(A_{n})\supseteq\Sob_{s}\,,\qquad A_{n}|\Sob_{s}\in\B(\Sob_{s},\H)\,,
$$
and
$$
\text{$A^{*}_{n}+A_{n}$ is $H$-bounded with relative bound $a<1$};
$$ 
further suppose, whenever $s=\frac12$, that $H$ is semi-bounded and $\mu=0$.\par
Let
$$
H_{n}:\Sob_{1}\subseteq\H\to\H\,,\qquad H_{n}:=H+A^{*}_{n}+A_{n}\,,
$$
$$
\widetilde H_{n}:\Sob_{1}\subseteq\H\to\H\,,\qquad \widetilde H_{n}:=H_{n}-A_{n}{R}A_{n}^{*}\,.
$$
If 
\be\label{Aenne}
\lim_{n\uparrow\infty}\|A_{n}-A\|_{\Sob_{1},\H}=0\,,
\ee
then
\be\label{wdn}
\lim_{n\uparrow\infty}\,\widetilde H_{n}=\widehat H_{0}\quad \text{in norm-resolvent sense.}
\ee
Let ${T}:\D(T)\subseteq\H\to\H$, $\D(T)\supseteq\D(\widehat H_{0})$, be symmetric and $\widehat H_{0}$-bounded with relative bound $\widehat a<1$; let $\widehat H_{{T}
}$ be the self-adjoint operator $\widehat H_{T}:=\overline{H}+A^{*}+A_{{T}}
$,  $\dom(\widehat H_{{T}
})=\dom(\widehat H_{0})$. If, alongside with \eqref{Aenne}, there exist a sequence   $\{E_{n}\}_{1}^{\infty}$ of bounded symmetric operators in $\H$
such that  
\be\label{Hy1}
\text{$A_{n}{R}A^{*}_{n}+E_{n}$ is $\widetilde H_{n}$-bounded with $n$-independent relative bound $\widetilde a<1$}
\ee
and  
\be\label{Hy2}
\lim_{n\uparrow\infty}\|A_{n}{R}A^{*}_{n}+E_{n}-{T}\|_{\D(T),\H}=0\,,
\ee
then
$$\lim_{n\uparrow\infty}\,(H_{n}+E_{n})=\widehat H_{{T}
}\quad \text{in norm-resolvent sense.}
$$
\end{theorem}
\begin{proof}  By Remark \ref{RR}, one has $\widehat H_{\mu}=H_{\mathbb{\Theta
}}$, where
$$
\mathbb{\Theta
}:=
\begin{bmatrix}-\mu&1-G^{*}\\1-G&-{R}
\end{bmatrix}\,.
$$
Let
$$
\mathbb{\Sigma}_{n}:{\H}\to\H\oplus\H\,,\qquad \mathbb{\Sigma}_{n}\psi:=A_{n}\psi\oplus\psi\,,
$$
and 
$$
\mathbb{\Theta
}_{n}:=
\begin{bmatrix}A_{n}{R}A^{*}_{n}-\mu&1\\1&0
\end{bmatrix}\,.
$$ 
Notice that, by $A_{n}\in\B(\Sob_{1/2},\H)$ and $R\in\B(\Sob_{-1/2},\Sob_{1/2})$, $A_{n}{R}A^{*}_{n}\in\B(\H)$; therefore $\Theta_{n}$ is bounded with bounded inverse given by 
\begin{align*}
\mathbb{\Lambda}_{n}:=\mathbb{\Theta
}_{n}^{-1}
=\begin{bmatrix}0&\uno\\\uno
&\mu-A_{n}{R}A^{*}_{n}
\end{bmatrix}
\end{align*}
and, by the Rellich-Kato theorem, $H+\mathbb{\Sigma}^{*}_{n}\mathbb{\Lambda}_{n}\mathbb{\Sigma}_{n}=\widetilde H_{n}+\mu$ is self-adjoint with domain $\dom(\widetilde H_{n})=\Sob_{1}$ ($A^{*}_{n}+A_{n}+A_{n}RA^{*}_{n}$ is symmetric since $A_{n}$ is closable).\par
If $0\le s<\frac12$, then, by \eqref{t-u} and
$$
\mathbb{\Sigma}_{n}R_{z}\mathbb{\Sigma}_{n}^{*}=
\begin{bmatrix}A_{n}{R_{z}}A^{*}_{n}&A_{n}R_{z}\\R_{z}A_{n}^{*}&R_{z}
\end{bmatrix}\,,
$$
one gets $\|\mathbb{\Lambda}_{n}\mathbb{\Sigma}_{n}R_{\pm i\gamma}\mathbb{\Sigma}_{n}^{*}\|_{\H\oplus\H, \H\oplus\H}\to 0$ as $|\gamma|\uparrow\infty$; so,  by Theorem \ref{TRK},  $\widetilde H_{n}+\mu$ has resolvent given by formula \eqref{K-RK1}. \par 
Suppose now $s=\frac12$, $H$ semi-bounded and $\mu=0$. Since $(A_{n}R)^{*}$ and $A_{n}R$ norm converge to $G$ and $G^{*}$ respectively and, by Remark \ref{prev},  $\uno-G$ and $\uno-G^{*}$ have bounded inverses whenever $\lambda_{\circ}$ is chosen sufficiently far away from $\sigma(H)$ (see next Remark \ref{LcL} as regard the freedom to choose the value of $\lambda_{\circ}$), $\uno-RA^{*}_{n}$ and $\uno-A_{n }R$ have bounded inverses as well whenever $n$ is sufficiently large. Hence, by the relation $$-\widetilde H_{n}+\lambda_{\circ}=(\uno-A_{n}R)(-H+\lambda_{\circ})(\uno-RA^{*}_{n})\,,
$$
one gets
\begin{align*}
&(-H+\lambda_{\circ} )^{-1}=(\uno-RA_{n}^{*})(-\widetilde H_{n}+\lambda_{\circ} )^{-1}(\uno-A_{n}R)\\
=&(-\widetilde H_{n}+\lambda_{\circ} )^{-1}-\begin{bmatrix}RA_{n}^{*}&{R}
\end{bmatrix}\begin{bmatrix}(-\widetilde H_{n}+\lambda_{\circ} )^{-1}
 &(\uno-RA_{n}^{*})^{-1}\\(\uno-A_{n}R)^{-1}&0
\end{bmatrix}\begin{bmatrix}A_{n}R\\{R}\end{bmatrix}\\
=&(-\widetilde H_{n}+\lambda_{\circ} )^{-1}-\begin{bmatrix}RA_{n}^{*}&{R}
\end{bmatrix}\begin{bmatrix}0
 &\uno-A_{n}R\\\uno-RA_{n}^{*}&{-R}
\end{bmatrix}^{-1}\begin{bmatrix}A_{n}R\\{R}
\end{bmatrix}\\
=&(-\widetilde H_{n}+\lambda_{\circ} )^{-1}-\begin{bmatrix}RA_{n}^{*}&{R}
\end{bmatrix}\left(\begin{bmatrix}A_{n}RA_{n}^{*}
 &\uno\\\uno&0
\end{bmatrix}-\begin{bmatrix}A_{n}RA_{n}^{*}
 &A_{n}R\\ RA_{n}^{*}&{R}
\end{bmatrix}\right)^{-1}\begin{bmatrix}A_{n}R\\{R}
\end{bmatrix}\\
=&(-\widetilde H_{n}+\lambda_{\circ} )^{-1}-R{\mathbb\Sigma}_{n} \left({\mathbb\Theta}_{n}-{\mathbb\Sigma}_{n} R{\mathbb\Sigma}_{n}^{*}\right)^{-1}{\mathbb\Sigma}_{n} R\,.
\end{align*}
This, together with \cite[Theorem 2.19 and Remark 2.20]{CFP}, gives the resolvent formula \eqref{K-RK1} for $\widetilde H_{n}$.  \par
Once we get formula \eqref{K-RK1} for $(-\widetilde H_{n}+z )^{-1}$ and for any $s\in[0,\frac12]$, since 
$$
\lim_{n\uparrow\infty}\|\mathbb{\Sigma}-\mathbb{\Sigma}_{n}\|_{\Sob_{1},\H\oplus\H}
=\lim_{n\uparrow\infty}\|A-A_{n}\|_{\Sob_{1},\H}=0\,, 
$$
$$
\lim_{n\uparrow\infty}\|(\mathbb{\Theta}_{n}-\mathbb{\Sigma}_{n}{R}\mathbb{\Sigma}_{n}^{*})-\mathbb{\Theta}\|_{\H\oplus\H,\H\oplus\H}
=\lim_{n\uparrow\infty}\left\|\begin{bmatrix}0&G^{*}-A_{n}{R}\\G-{R}A^{*}_{n}
&0
\end{bmatrix}\right\|_{\H\oplus\H,\H\oplus\H}=0\,,
$$
and $\dom({\mathbb{\Theta}}_{n})=\dom({\mathbb{\Theta}})=\H\oplus\H$,  by Theorem \ref{conv}, one gets
$$\lim_{n\uparrow\infty}\,(\widetilde H_{n}+\mu
)=\lim_{n\uparrow\infty}\,(H+\mathbb{\Sigma}^{*}_{n}\mathbb{\Lambda}_{n}\mathbb{\Sigma}_{n})=H_{\mathbb{\Theta
}}=\widehat H_{\mu}\quad \text{in norm-resolvent sense.}
$$
Equivalently, 
\be\label{NC}
\lim_{n\uparrow\infty}\,\widetilde H_{n}=\widehat H_{0}\quad \text{in norm-resolvent sense.}
\ee
Now, let us consider the relations, which hold for $\gamma\in\RE$, $|\gamma|$ sufficiently large,  
$$
(-(H_{n}+E_{n})\pm i\gamma)^{-1}=(-(\widetilde H_{n}+T_{n})\pm i\gamma)^{-1}=(\uno-(-\widetilde H_{n}\pm i\gamma)^{-1}T_{n})^{-1}(-\widetilde H_{n}\pm i\gamma)^{-1}\,,
$$
where $T_{n}:=A_{n}{R}A^{*}_{n}+E_{n}$, and
$$
(-(\widehat H_{0}+T)\pm i\gamma)^{-1}=(-\widehat H_{0}\pm i\gamma)^{-1}(\uno-T(-\widehat H_{0}\pm i\gamma)^{-1})^{-1}\,.
$$
We also use the relation
$$
(-\widetilde H_{n}\pm i\gamma)^{-1}-(\widehat H_{0}\pm i\gamma)^{-1}\\
=\lfloor(-\widetilde H_{n}\pm i\gamma)^{-1}\widetilde H_{n}\rceil(-\widehat H_{0}\pm i\gamma)^{-1}
-(-\widetilde H_{n}\pm i\gamma)^{-1}\widehat H_{0}(-\widehat H_{0}\pm i\gamma)^{-1}
$$
(here and below we use the brackets $\lfloor...\rceil$ to group maps which provide bounded operators defined on the whole $\H$). Therefore one gets 
\begin{align*}
&(-(H_{n}+E_{n})\pm i\gamma)^{-1}-(-(\widehat H_{0}+T)\pm i\gamma)^{-1}\\
=&\lfloor (-(\widetilde H_{n}+T_{n})\pm i\gamma)^{-1}(\widetilde H_{n}+T_{n})\rceil  (-(\widehat H_{0}+T)\pm i\gamma)^{-1}\\
-&(-(\widetilde H_{n}+T_{n})\pm i\gamma)^{-1}(\widehat H_{0}+T)(-(\widehat H_{0}+T)\pm i\gamma)^{-1}\\
=&(\uno-(-\widetilde H_{n}\pm i\gamma)^{-1}T_{n})^{-1}\lfloor (-\widetilde H_{n}\pm i\gamma)^{-1}(\widetilde H_{n}+T_{n})\rceil  (-\widehat H_{0}\pm i\gamma)^{-1}(\uno-T(-\widehat H_{0}\pm i\gamma)^{-1})^{-1}\\
-&(\uno-(-\widetilde H_{n}\pm i\gamma)^{-1}T_{n})^{-1}(-\widetilde H_{n}\pm i\gamma)^{-1}(\widehat H_{0}+T)(-\widehat H_{0}\pm i\gamma)^{-1}(\uno-T(-\widehat H_{0}\pm i\gamma)^{-1})^{-1}\\
=&(\uno-(-\widetilde H_{n}\pm i\gamma)^{-1}T_{n})^{-1}
\big((-\widetilde H_{n}\pm i\gamma)^{-1}-(-\widehat H_{0}\pm i\gamma)^{-1}\big)
(\uno-T(-\widehat H_{0}\pm i\gamma)^{-1})^{-1}\\
&+(-(\widetilde H_{n}+T_{n})\pm i\gamma)^{-1}(T_{n}-T)(-(\widehat H_{0}+T)\pm i\gamma)^{-1}
\end{align*}
and so,
\begin{align*}
&\|(-(H_{n}+E_{n})\pm i\gamma )^{-1}-(-(\widehat H_{0}+T)\pm i\gamma )^{-1}\|_{\H, \H}\\
\le& \|(\uno-T(-\widehat H_{0}\pm i\gamma )^{-1})^{-1}\|_{\H, \H}\| (\uno-(-\widetilde H_{n}\pm i\gamma )^{-1}T_{n})^{-1}\|_{\H, \H}\,\times\\
&\times\|(-\widetilde H_{n}\pm i\gamma )^{-1}-(-\widehat H_{0}\pm i\gamma )^{-1} \|_{\H, \H}\\
&+\frac1{|\gamma|}\ \|(T_{n}-T)(-(\widehat H_{0}+T)\pm i\gamma )^{-1}\|_{\H, \H}\,.
\end{align*}
By \eqref{Hy1},
$$
\sup_{n\ge 1}\|(\uno-(-\widetilde H_{n}\pm i\gamma )^{-1}T_{n})^{-1}\|_{\H, \H}\le \frac1{1-\widetilde a}
$$
and, since $T$ is $\widehat H_{0}$-bounded,
$$
\|(-\widehat H_{0}\pm i\gamma )^{-1}\|_{\H, \D(T)}\le 
\|T(-\widehat H_{0}\pm i\gamma )^{-1}\|_{\H, \H}+
(\|(-\widehat H_{0}\pm i\gamma )^{-1}\|_{\H, \H}<+\infty\,.
$$
Then, by \eqref{Hy2},
\begin{align*}
&\lim_{n\uparrow\infty}\|(T_{n}-T)(-(\widehat H_{0}+T)\pm i\gamma )^{-1}\|_{\H, \H}\\
\le& \|(-(\widehat H_{0}+T)\pm i\gamma )^{-1}\|_{\H, \D(T)}\lim_{n\uparrow\infty}\|T_{n}-T\|_{\D(T),\H}\\
\le& 
\|(\uno-T(-\widehat H_{0}\pm i\gamma )^{-1})^{-1}\|_{\H, \H}
\|(-\widehat H_{0}\pm i\gamma )^{-1}\|_{\H, \D(T)}
\lim_{n\uparrow\infty}\|T_{n}-T\|_{\D(T),\H}
=0\,.
\end{align*}
Hence,  by \eqref{NC}, the sequence $H_{n}+E_{n}$  converges in norm-resolvent sense to $\widehat H_{T}$  as $n\uparrow\infty$.
\end{proof}
\begin{remark}\label{suit} Previous Theorem \ref{bb} suggests that if the sequence $A_{n}{R}A_{n}^{*}$ were convergent then one could take $E_{n}=0$ and $T=AG\equiv A{R}A^{*}$. However $A{R}A^{*}$ is ill-defined in presence of strongly singular interactions and  $E_{n}$'s role is to compensate the divergence of $A_{n}{R}A_{n}^{*}$ as $n\to+\infty$ so that $A_{n}{R}A_{n}^{*}+E_{n}$ converges to some regularized version of $A{R}A^{*}$; see next subsection for the case of quantum fields models.
\end{remark}
\begin{remark}\label{Hz} Suppose that the operator $\widehat H_{0}$ is self-adjoint with resolvent given by \eqref{res-c} for some $\mu\in\RE$ and let $A_{n}\in\B(\H)$ defined by $A_{n}:=niAR_{ni}$, where $A\in\B(\Sob_{1},\H)$ satisfies the hypotheses in Lemma  \ref{sl}. Since $R_{z}A^{*}_{n}$ and $A_{n }R_{z}$ norm converge to $G_{z}$ and $G^{*}_{{\bar z} }$ respectively and since $\uno-G_{\pm i\gamma}$ and $\uno-G^{*}_{\mp i\gamma} $ have bounded inverses whenever  $|\gamma| \gg 1$ (see Lemma \ref{sl}), $\uno-R_{\pm i\gamma}A^{*}_{n}$ and $\uno-A_{n }R_{\pm i\gamma}$ have bounded inverses as well whenever $n$ is sufficiently large; moreover $(\uno-R_{\pm i\gamma}A^{*}_{n})^{-1}$ and $(\uno-A_{n }R_{\pm i\gamma})^{-1}$ norm converge to $(\uno-G_{\pm i\gamma})^{-1}$ and $(\uno-G^{*}_{\mp i\gamma} )^{-1}$ respectively. Hence 
\begin{align}\label{He}
&\lim_{n\uparrow\infty}\|(\uno-R_{\pm i\gamma}A^{*}_{n})^{-1}R_{\pm i\gamma }(\uno-A_{n}R_{\pm i\gamma})^{-1}-
(\uno-G_{\pm i\gamma})^{-1}R_{\pm i\gamma }(\uno-G^{*}_{\mp i\gamma} )^{-1}\|_{\H, \H}=0\,.
\end{align}
Since
$$
(\uno-A_{n}R_{z})(-H+z)(\uno-R_{z}A^{*}_{n})=
(-\widetilde H_{n}+z)+(\lambda_{\circ} -z)A_{n}{R}R_{z}A^{*}_{n}\,,
$$
one has 
$$
(-\widetilde H_{n}\pm i\gamma )^{-1}=\left((\uno-A_{n}R_{\pm i\gamma})(-H+z)(\uno-R_{\pm i\gamma} A^{*}_{n})+(\pm i\gamma-\lambda_{\circ} )A_{n}{R}R_{\pm i\gamma }A^{*}_{n}\right)^{-1}
$$
and so, by \eqref{He} and \eqref{wdn}, one gets 
$$
(-\widehat H_{0}\pm i\gamma )^{-1}=\left((\uno-G^{*}_{\mp i\gamma} )(-H\pm i\gamma )(\uno-G_{\pm i\gamma} )+(\pm i\gamma-\lambda_{\circ})G^{*}G_{\pm i\gamma} \right)^{-1}\,.
$$
Hence 
$$
-\widehat H_{0}\pm i\gamma =(\uno-G^{*}_{\mp i\gamma} )(-H\pm i\gamma )(\uno-G_{\pm i\gamma} )+(\pm i\gamma-\lambda_{\circ})G^{*}G_{\pm i\gamma}
$$
which, by \eqref{RG}, is equivalent to (compare with \cite[equation (15)]{LS})
\be\label{Rl}
-\widehat H_{0}+\lambda_{\circ} =(\uno-G^{*})(-H+\lambda_{\circ} )(\uno-G)\,.
\ee
\end{remark}
\vskip8pt
Our next aim is to show that the two resolvent formulae \eqref{res-T} and  \eqref{KF} (equivalently \eqref{Kbb}) coincide. At first, let us come back to Remark \ref{RR}: the map $\mathbb{\Sigma}\psi:=A\psi\oplus\psi$ there obviously belongs to $\B(\Sob_{1},\H\oplus\Sob_{1})$; hence, using the $\Sob_{-1}$-$\Sob_{1}$ duality induced by the dense embeddings $\Sob_{1}\hookrightarrow \H \hookrightarrow \Sob_{-1}$ (i.e., by the pairing $\langle\cdot,\cdot\rangle_{-1,1}$ defined in \eqref{pairing}), one gets the bounded operator
$$
\mathbb{G}_{z}:\H\oplus\Sob_{-1}\to\H\,,\qquad \mathbb{G}_{z}:=(\mathbb{\Sigma}R_{{\bar z} }\!)^{*}\,.
$$
This also gives $\mathbb{G}^{*}_{z}\in\B(\H,\H\oplus\Sob_{1})$ and so \eqref{Kbb} is well defined whenever  $(\mathbb{\Sigma}_{\mathbb{\Theta
}_{T}}\mathbb{G}_{z})^{-1}\in\B(\H\oplus\Sob_{1},\H\oplus\Sob_{-1})$, where now 
${\mathbb{\Theta}_{T}}$ is regarded as an operator from  $\H\oplus\Sob_{-1}$ to $\H\oplus\Sob_{1}$. Let us remark that in this setting \eqref{Kbb} still conforms with the framework in \cite{P01} (also see \cite{JMPA}, \cite{CFP}); indeed there the dual couple $\X^{*}$-$\X$ (here given by $(\H\oplus\Sob_{-1})$-$(\H\oplus\Sob_{1})$) comes into play. \par 
Because, by \eqref{Rl}, $(\uno-G)\phi\in\Sob_{1}$ whenever $\psi\in\D(\widehat H_{0})$ and supposing that $\D(\widehat H_{0})\subseteq\D(T)$, the block operator matrix
\be\label{TT}
{\mathbb{\Theta}_{T}}=\begin{bmatrix}-{T}
 &\uno-G^{*}\\\uno-G&-{R}
\end{bmatrix}:\D(\widehat H_{0})\oplus\H\subseteq\H\oplus\Sob_{-1}\to\H\oplus\Sob_{1}
\ee
is well defined. Analogously 
$$
\mathbb{\Sigma}_{\mathbb{\Theta
}_{T}}\mathbb{G}_{z}
=\begin{bmatrix} A_{T}G_{z}
 &G^{*}_{{\bar z} }-\uno\\G_{z}-\uno&{R_{z}}
\end{bmatrix}: \D(\widehat H_{0})\oplus\H\subseteq\H\oplus\Sob_{-1}\to\H\oplus\Sob_{1}
$$
is well defined as well. Since the unbounded operator $-R:\H\subseteq\Sob_{-1}\to\Sob_{1}$ has the bounded inverse $H-\lambda_{\circ} :\Sob_{1}\to\H$,  by \eqref{Rl} and the first Schur complement, the candidate for the inverse of ${\mathbb{\Theta}_{T}}$ is  
\begin{align}\label{IT}
&\begin{bmatrix} \wR_{T}
 &\wR_{T}(-\widehat H_{0}+\lambda_{\circ} )(\uno-G)^{-1}\\ (\uno-G^{*})^{-1}(-\widehat H_{0}+\lambda_{\circ} )\wR_{T}&(\uno-G^{*})^{-1}(-\widehat H_{0}+\lambda_{\circ} )\wR_{T}(-\widehat H_{0}+\lambda_{\circ} )(\uno-G)^{-1}-(-H+\lambda_{\circ} )
\end{bmatrix}
\nonumber
\\
=&\begin{bmatrix} \wR_{0}(\uno-T\wR_{0})^{-1}
 &(\uno-\wR_{0}T)^{-1}(\uno-G)^{-1}\\ (\uno-G^{*})^{-1}(\uno-T\wR_{0})^{-1}&\big((\uno-G^{*})^{-1}(\uno-T\wR_{0})^{-1}(\uno-G^{*})-\uno\big)(-H+\lambda_{\circ} )
\end{bmatrix}\,,
\end{align}
where for brevity we set 
$$\wR_{T}:=(-\widehat H_{T}+\lambda_{\circ} )^{-1}
=\wR_{0}(\uno-T\wR_{0})^{-1}=(\uno-\wR_{0}T)^{-1}\wR_{0}\,.
$$
In the following, as regards ${\mathbb \Theta}_{T}$, we use the notion of self-adjointness for a linear operator acting between dual pairs: a densely defined $L:\dom(L)\subseteq\H\oplus\Sob_{-1}\to\H\oplus\Sob_{1}$ is said to be self-adjoint whenever $L^{*}=L$, where $L^{*}:\dom(L^{*})\subseteq\H\oplus\Sob_{-1}\to\H\oplus\Sob_{1}$ is the dual of $L$ with respect to the duality induced by the pairing $\langle\cdot,\cdot\rangle_{-1,1}$. 
\vskip8pt
If $H$ is semibounded then Theorem \ref{alo} conforms with Theorem \ref{resolvent}:
\begin{theorem}\label{sb} Let  $H$ be semibounded; let $A\in\B(\Sob_{s},\H)$ for some $0<s<1$ and such that both $\ker(A|\Sob_{1})$ and $\ran(A|\Sob_{1})$ are dense in $\H$. 
Then $$\widehat H_{0}:=\overline{H}+A^{*}+A_{0}$$ is self-adjoint, semibounded with 
domain 
$$\D(\widehat H_{0})=\{\psi\in\Sob_{1-s}:\psi-G\psi\in\Sob_{1}\}$$
and resolvent given, for any $z\in\varrho(H)\cap \varrho (\widehat H_{0})$, by 
\be\label{res-sb}
(-\widehat H_{0}+z)^{-1}
=(-H+z)^{-1}-\begin{bmatrix}G_{z}&R_{z}\end{bmatrix}\begin{bmatrix}A_{0} G_{z}&G^{*}_{{\bar z} }-\uno\\G_{z}-\uno&R_{z}
\end{bmatrix}^{-1}\begin{bmatrix}G_{{\bar z} }^{*}\\R_{z}\end{bmatrix}\,.
\ee
If ${T}:\D(T)\subseteq\H\to\H$, $\D(T)\supseteq\D(\widehat H_{0})$, is symmetric and $\widehat H_{0}$-bounded with relative bound ${\widehat a}<1$ then $$\widehat H_{{T}
}:=\overline{H}+A^{*}+A_{{T}}$$ is self-adjoint and semibounded, with domain 
$\D(\widehat H_{{T}
})=\D(\widehat H_{0})$ 
and resolvent given, for any $z\in\varrho(H)\cap \varrho (\widehat H_{T})$, by
\begin{align}\label{KF-sb}
(-\widehat H_{{T}
}+z)^{-1}
=&(-H+z)^{-1}-
\begin{bmatrix}G_{z}&R_{z}\end{bmatrix}\begin{bmatrix}A_{{T}
} G_{z}&G^{*}_{{\bar z} }-\uno\\G_{z}-\uno&R_{z}
\end{bmatrix}^{\!-1}\begin{bmatrix}G_{{\bar z} }^{*}\\R_{z}\end{bmatrix}\,.
\end{align}
The block operator matrix inverse in \eqref{KF-sb} exists as a bounded operator from $\H\oplus\Sob_{1}$ to $\H\oplus\Sob_{-1}$; whenever $T=0$,  the same inverse exists also as a bounded operator in $\H\oplus\H$.
\end{theorem}
\begin{proof}
By \cite[Theorem 2.2.18]{Tre}, ${\mathbb \Theta}_{T}$ defined in \eqref{TT} it closed and, by \cite[Theorem 2.3.3]{Tre}, it has a bounded inverse ${\mathbb \Theta}_{T}^{-1}$ given by the block operator matrix in \eqref{IT}. \par
If $H$ is semibounded then, by \eqref{Rl},  $\widehat H_{0}$ and hence (by Rellich-Kato theorem) $\widehat H_{T}$ are  semibounded as well. \par 
By Remark \ref{prev} and \eqref{Rl}, taking $\lambda_{\circ} \in\RE$ sufficiently far away from $\sigma(H)$ in the definition \eqref{GL}, one has $\lambda_{\circ} \in\varrho(H)\cap\varrho(\widehat H_{0})$ and 
\be\label{resLS}
(-\widehat H_{0}+\lambda_{\circ} )^{-1}=(\uno-G)^{-1}(-H+\lambda_{\circ} )^{-1}(\uno-G^{*})^{-1}\,,
\ee
i.e.,
\begin{align*}
&(-H+\lambda_{\circ} )^{-1}=(\uno-G)(-\widehat H_{0}+\lambda_{\circ} )^{-1}(\uno-G^{*})\\
=&(-\widehat H_{0}+\lambda_{\circ} )^{-1}-\begin{bmatrix}G&{R}
\end{bmatrix}\begin{bmatrix}(-\widehat H_{0}+\lambda_{\circ} )^{-1}
 &(\uno-G)^{-1}\\(\uno-G^{*})^{-1}&0
\end{bmatrix}\begin{bmatrix}G^{*}\\{R}\end{bmatrix}\\
=&(-\widehat H_{0}+\lambda_{\circ} )^{-1}-\begin{bmatrix}G&{R}
\end{bmatrix}\begin{bmatrix}0
 &\uno-G^{*}\\\uno-G&{-R}
\end{bmatrix}^{-1}\begin{bmatrix}G^{*}\\{R}
\end{bmatrix}\,.
\end{align*}
This gives the resolvent formula 
\be\label{res-L}
(-\widehat H_{0}+\lambda_{\circ} )^{-1}=(-H+\lambda_{\circ} )^{-1}+\mathbb{G}\,\mathbb{\Theta
}_{0}^{-1}\mathbb{G}^{*}\,.
\ee
Therefore ${\mathbb G}\,{\mathbb \Theta}_{0}^{-1}{\mathbb G}^{*}=\widehat R_{0}-R$ is symmetric. By $\ran(A|\Sob_{1})$ dense, $\ran(G^{*})=\ran(AR)$ is dense and so $\ran({\mathbb G}^{*})=\ran(G^{*})\oplus\Sob_{1}$ is dense as well. Thus ${\mathbb \Theta}_{0}^{-1}$ is symmetric (both as an operator in $\H\oplus\H$ and as an operator from $\H\oplus\Sob_{1}$ to $\H\oplus\Sob_{-1}$); hence it is self-adjoint since bounded. So, by \cite[Theorem 5.30, Chap. III]{K}, ${\mathbb \Theta}_{0}$ is self-adjoint. Then, since ${\mathbb \Sigma}_{{\mathbb \Theta}_{0}}{\mathbb G}_{z}=-({\mathbb \Theta}_{0}+{\mathbb\Sigma}({\mathbb G}-{\mathbb G}_{z}))$ and $({\mathbb\Sigma}({\mathbb G}-{\mathbb G}_{z}))^{*}=
{\mathbb\Sigma}({\mathbb G}-{\mathbb G}_{{\bar z} })$, by \cite[Theorem 5.30, Chap. III]{K} again,$(({\mathbb \Sigma}_{{\mathbb \Theta}_{0}}{\mathbb G}_{z}^{-1})^{*}=({\mathbb \Sigma}_{{\mathbb \Theta}_{0}}{\mathbb G}_{{\bar z} })^{-1}$ for any complex conjugate couple for which the inverses exist.
Therefore, by \cite[Theorem 2.19 and Remark 2.20]{CFP}, the existence of the bounded inverse ${\mathbb \Theta}_{0}^{-1}$ implies that the resolvent formula 
$$
(-\widehat H_{0}+z)^{-1}=(-H+z)^{-1}-\mathbb{G}_{z}\,({\mathbb \Sigma}_{{\mathbb\Theta}_{0}}{\mathbb G_{z}})^{-1}\mathbb{G}_{{\bar z} }^{*}
$$
holds for any $z\in \varrho(H)\cap\varrho(\widehat H_{0})$. The latter is equivalent to 
\eqref{res-sb}.\par
By the same kind of reasonings as above,  to prove the resolvent formula \eqref{KF-sb} it suffices to show that it holds in the case $z=\lambda_{\circ} $, i.e., that 
\be\label{tt}
(-\widehat H_{T}+\lambda_{\circ} )^{-1}=R+{\mathbb G}\,{\mathbb \Theta}_{T}^{-1}{\mathbb G}^{*}\,.
\ee
By \eqref{IT}, \eqref{Rl} and \eqref{res-T}, one gets
\begin{align*}
&R+{\mathbb G}\,{\mathbb \Theta}_{T}^{-1}{\mathbb G}^{*}=R-
\begin{bmatrix}G&R\end{bmatrix}\begin{bmatrix} 0&0\\0&-H+\lambda_{\circ} 
\end{bmatrix}\begin{bmatrix}G^{*}\\R\end{bmatrix}\\
&+\begin{bmatrix}G&R\end{bmatrix}\begin{bmatrix} \wR_{0}(\uno-T\wR_{0})^{-1}
 &(\uno-\wR_{0}T)^{-1}(\uno-G)^{-1}\\ (\uno-G^{*})^{-1}(\uno-T\wR_{0})^{-1}&(\uno-G^{*})^{-1}(\uno-T\wR_{0})^{-1}(\uno-G^{*})(-H+\lambda_{\circ} )
\end{bmatrix}\begin{bmatrix}G^{*}\\R\end{bmatrix}\\
=&\begin{bmatrix}G&R\end{bmatrix}\begin{bmatrix} \wR_{0}(\uno-T\wR_{0})^{-1}
 &(\uno-\wR_{0}T)^{-1}\wR_{0}\\ (\uno-G^{*})^{-1}(\uno-T\wR_{0})^{-1}&(\uno-G^{*})^{-1}(\uno-T\wR_{0})^{-1}
\end{bmatrix}\begin{bmatrix}G^{*}\\1-G^{*}\end{bmatrix}\\
=&\begin{bmatrix}G&R\end{bmatrix}\begin{bmatrix} \wR_{0}(\uno-T\wR_{0})^{-1}\\(\uno-G^{*})^{-1}(\uno-T\wR_{0})^{-1}\end{bmatrix}=\big( G\wR_{0}+R(\uno-G^{*})^{-1}\big)(\uno-T\wR_{0})^{-1}\\
=&\wR_{0}(\uno-T\wR_{0})^{-1}=
(-\widehat H_{0}+z)^{-1}+(-\widehat H_{0}+z)^{-1}(1-{T}
(-\widehat H_{0}+z)^{-1})^{-1}{T}
(-\widehat H_{0}+z)^{-1}\\
=&(-\widehat H_{T}+\lambda_{\circ} )^{-1}\,.
\end{align*}
\end{proof}
\begin{remark}\label{LcL} In the proof of Theorem \ref{sb} we took $|\lambda_{\circ}|$ sufficiently large so that $\uno-G$ has a bounded inverse; hence the validity such a theorem seems to depend on the value of  $\lambda_{\circ}$. This is not the case. Indeed, given $\lambda$ such that $\uno-G_{\lambda}$ has a bounded inverse, let $A^{(\lambda)}_{T}G_{z}:=T+A(G_{\lambda}-G_{z})$; then 
$A^{(\lambda)}_{T}G_{z}=A_{T+T_{\lambda}}G_{z}$, where $T_{\lambda}:=A(G-G_{\lambda})=(\lambda-\lambda_{\circ})G^{*}G_{\lambda}=(\lambda-\lambda_{\circ})G_{\lambda}^{*}G\in\B(\H)$ and so, by Theorem \ref{sb}, 
$\widehat H^{(\lambda)}_{T}:=\overline H+A^{*}+A_{T+T_{\lambda}}$ is self-adjoint with domain $\dom(\widehat H_{T+T_{\lambda}})=\{\psi\in\Sob_{1-s}:\psi-G_{\lambda}\psi\in\Sob_{1}\}$ and resolvent given by formula \eqref{KF-sb} with $T$ replaced by $T+T_{\lambda}$. Since $\ran(G-G_{\lambda})\subseteq\Sob_{1}$, $\dom(\widehat H_{T+T_{\lambda}})=\{\psi\in\Sob_{1-s}:\psi-G\psi\in\Sob_{1}\}$ and, since $T_{\lambda}$ is symmetric and bounded, $\widehat H_{T}:=\widehat H^{(\lambda)}_{T}-T_{\lambda}=\overline H+A^{*}+A_{0}+T=\overline H+A^{*}+A_{T}$ is self-adjoint with $\dom(\widehat H_{T})=\dom(\widehat H_{T+T_{\lambda}})$. By $\widehat H_{T}=\widehat H^{(\lambda)}_{0}+(T-T_{\lambda})$ and by Theorem \ref{sb}, the resolvent of $\widehat H_{T}$ is given by formula \eqref{KF-sb} with $A_{T}G_{z}$ replaced by $A^{(\lambda)}_{T-T_{\lambda}}G_{z}$, i.e., since  $A^{(\lambda)}_{T-T_{\lambda}}G_{z}=A_{T}G_{z}$, by \eqref{KF-sb} itself.
\par Similar considerations apply to the proof of Theorem \ref{bb} (in the case $s=\frac12$). Defining $\widetilde H^{(\lambda)}_{n}:=H_{n}-A_{n}R_{\lambda}A_{n}^{*}$ with $|\lambda|$ so large that $\uno-G_{\lambda}$ has a bounded inverse, by Theorem \ref{bb}, $\widetilde H^{(\lambda)}_{n}$ converges in norm resolvent sense to $\widehat H^{(\lambda)}_{0}$. Then, since $\widetilde H_{n}=\widetilde H^{(\lambda)}_{n}-A_{n}RA_{n}^{*}+A_{n}R_{\lambda}A_{n}^{*}=\widetilde H^{(\lambda)}_{n}-(\lambda-\lambda_{\circ})A_{n}R(A_{n}R_{\lambda})^{*}$, $\widehat H_{0}=\widehat H^{(\lambda)}_{0}-T_{\lambda}$ and $\|(\lambda-\lambda_{\circ})A_{n}R(A_{n}R_{\lambda})^{*}-T_{\lambda}\|_{\H,\H}\to 0$, one gets that $\widetilde H_{n}$ converges in norm resolvent sense to $\widehat H_{0}$ regardeless of the value of $\lambda_{\circ}$.
\end{remark}
\begin{remark}\label{kind} By the same kind of reasonings as in the proof of Theorem \ref{sb}, if $\widehat H_{T}$ is a self-adjoint operator with a resolvent given by \eqref{Kbb} with $(\mathbb{\Sigma}_{\mathbb{\Theta}_{T}}\mathbb{G}_{z})^{-1}\in\B(\H\oplus\Sob_{1},\H\oplus\Sob_{-1})$ and  $\ran (A|\Sob_{1})$ dense, 
then, by  $\big(\mathbb{G}_{z}(\mathbb{\Sigma}_{\mathbb{\Theta}_{T}}\mathbb{G}_{z})^{-1}\mathbb{G}^{*}_{{\bar z} }\big)^{*}=
\mathbb{G}_{{\bar z} }(\mathbb{\Sigma}_{\mathbb{\Theta}_{T}}\mathbb{G}_{{\bar z} })^{-1}\mathbb{G}^{*}_{z}$, by ${\mathbb \Theta}_{T}=-({\mathbb \Sigma}_{{\mathbb \Theta}_{T}}{\mathbb G}_{z}+{\mathbb\Sigma}({\mathbb G}-{\mathbb G}_{z}))$ and $({\mathbb\Sigma}({\mathbb G}-{\mathbb G}_{z}))^{*}=
{\mathbb\Sigma}({\mathbb G}-{\mathbb G}_{{\bar z} })$, one infers that $\mathbb{\Theta}_{T}$ is self-adjoint.
\end{remark}
\begin{remark}\label{mer} Regarding Theorem \ref{sb}, people working in extension theory could be puzzled by the fact that the family of self-adjoint operators $\widehat H_{T}$, coming out from the self-adjoint extensions of the symmetric $S=H|\ker(A)$, is parameterized by symmetric ($\widehat H_{0}$-bounded) operators $T$ which, unlike what is requested in Theorem \ref{resolvent}, are not necessarily self-adjoint. However, looking at the Kre\u\i n-type resolvent formula \eqref{KF-sb} (equivalently \eqref{Kbb}), the true parameterizing operator turns out to be ${\mathbb\Theta}_{T}$ in \eqref{TT} which is self-adjoint (relatively to the dual couple $\H\oplus\Sob_{-1}\,$-$\,\H\oplus\Sob_{1}$), even when $T$ is merely symmetric (see Remark \ref{kind}).  
\end{remark}
\begin{remark} Notice that, unlike Theorem \ref{resolvent}, in Theorem \ref{sb} one does not need  $A_{T}G_{z}$ to have a bounded inverse, i.e., one does not need hypothesis \eqref{H1}. Indeed, in \eqref{KF} (equivalently \eqref{Kbb}) the inverse $(\mathbb{\Sigma}_{\mathbb{\Theta}_{T}}\mathbb{G}_{z})^{-1}$ is regarded as an operator in $\H\oplus\H$ and so, since $R_{z}:\H\to\H$ has no bounded inverse, one uses the second Schur complement, which requires  $(A_{T}G_{z})^{-1}\in\B(\H)$; on the contrary, in \eqref{KF-sb} the same inverse block operator matrix is regarded as an operator from $\H\oplus\Sob_{1}$ to $\H\oplus\Sob_{-1}$ and so, since $R_{z}:\Sob_{-1}\subseteq\H\to\Sob_{1}$ has a bounded inverse, one can use the first Schur complement.  Also notice that, by \eqref{IT}, the only case where one can show that $(\mathbb{\Sigma}_{\mathbb{\Theta}_{T}}\mathbb{G}_{z})^{-1}\in\B(\H\oplus\H)$ without requiring   
$(A_{T}G_{z})^{-1}\in\B(\H)$ is the one given by  the choice $T=0$.
\end{remark}
\begin{remark}\label{BF} The strategy employed in Theorem \ref{sb} can be also applied to cases where $T$ is not $\widehat H_{0}$-bounded. For example, one can consider the case where $T=T_{1}+T_{2}$, with $T_{1}$ such that $H_{(1)}:=H+T_{1}$ is self-adjoint semibounded and $T_{2}$ is $\widehat H_{(1)}$-bounded with relative bound less that one, where $\widehat H_{(1)}$  is constructed in the same way as $\widehat H_{0}$, replacing  $H$ with $H_{(1)}$. This is what was done for the QFT model studied in \cite{L2}. If $A$ and $T_{1}$ self-adjoint satisfy the hypotheses in Corollary  \ref{c-sl}, and $T_{2}$ is $\widehat H_{T_{1}}$-bounded with relative bound less that one, then $\widehat H_{T}=\widehat H_{T_{1}}+T_{2}$ is self-adjoint with domain $\D(\widehat H_{T})=\{\psi\in \Sob_{1-s}:\psi-G\psi\in\Sob_{1}\}$.
\end{remark}
\begin{subsection}{Renormalizable QFT models.}\label{QFT}
Here we show, using results contained in \cite{LS}, how the 3-D Nelson model \cite{N} fits to our abstract framework; similar consideration apply to the other renormalizable models considered in \cite{LS} (2-D polaron-type model with point interactions), \cite{Sch1} (the 3-D Eckmann  and 2-D Gross models), \cite{Sch2} (the massless 3-D Nelson model) and \cite{L2} (the Bogoliubov-Fr\"ohlich model). \par
We take 
\be\label{Fock}
\H=L^{2}(\RE^{3N})\otimes\H_{b}\equiv\op\left(L^{2}(\RE^{3N})\otimes L^{2}_{sym}(\RE^{3n})\right)\,,
\ee
where $\H_{b}:=\Gamma_{\! b}(L^{2}(\RE^{3}))$ denotes the boson Fock space over $L^{2}(\RE^{3})$,
and $H=H_{\rm free}$, where $H_{\rm free}$ is the positive self-adjoint operator 
$$
H_{\rm free}:=-\Delta_{(3N)}\otimes \uno+\uno\otimes {\rm d}\Gamma_{\! b}\big((-\Delta_{(3)}+1)^{1/2}\big)\,.
$$
Here $\Delta_{(d)}:H^{2}(\RE^{d})\subseteq L^{2}(\RE^{d})\to L^{2}(\RE^{d})$ denotes the Laplace operator in $L^{2}(\RE^{d})$ with self-adjointness domain the Sobolev space $H^{2}(\RE^{d})$ and ${\rm d}\Gamma_{\! b}(L)$ denotes the boson second quantization of $L$ (see, e.g., \cite[Chapter 5]{Arai}). Since $0\in\varrho(H_{\rm free})$, we can take $\lambda_{\circ} =0$ in the definition \eqref{GL} of $G$, so that $G=-(AH_{\rm free}^{-1})^{*}$. In order to define the appropriate annihilation operator $A$ we use the identification $L^{2}(\RE^{3N})\otimes\H_{b}\equiv L^{2}(\RE^{3N};\H_{b})$
which maps $\psi\otimes\Phi$ to $\x\mapsto\Psi(\x):=\psi(\x)\Phi$. Given $v:=(-\Delta_{(3)}+1)^{-1/4}\delta$, $\delta$ denoting the Dirac delta distribution supported at $0\in\RE^{3}$,  we define
\be\label{ann}
(A\Psi)(\x):=a(v_{\x})\Psi(\x)\,,
\ee
where $$v_{\x}(y):={g}\, \sum_{j=1}^{N}v(x_{j}-y)\,,\qquad g\in\RE\,,\quad\x\equiv(x_{1},\dots,x_{N})\in\RE^{3N}$$ 
and 
$a(v)$ denotes the annihilation operator in $\H_{b}$ with test vector $v$ (see, e.g. \cite[Section 5.7]{Arai}). By \cite[Corollary 3.2]{LS}, one has $G\in\B(L^{2}(\RE^{3N})\otimes\H_{b},\dom(H_{\rm free}^{s}))$ for any $s<\frac14$, equivalently 
$A\in\B(\dom(H_{\rm free}^{s}),L^{2}(\RE^{3N})\otimes\H_{b})$ for any $s>\frac34$. By \cite[Lemma 2.2]{LS}, $\ker(A|\dom(H_{\rm free}))$ is dense in $L^{2}(\RE^{3N};\H_{b})$.  The proof that $\ran(A|\dom(H_{\rm free}))$ is dense in $L^{2}(\RE^{3N};\H_{b})$ follows the same kind of reasonings as in the proof of \cite[Lemma 2.2]{LS}: let $D_{\x}\subset H^{2}(\RE^{3})$ be a $L^{2}(\RE^{3})$-dense set of $\x$-smooth functions $f_{\x}$ such that $\langle v_{\x},f_{\x}\rangle\not=0$, where  $\langle\cdot,\cdot\rangle$ denotes the $H^{-2}(\RE^{3})$-$H^{2}(\RE^{3})$ duality; then, given $\Phi(f_{\x})$, the coherent state generated by $f_{\x}$, one uses the relation $a(v_{\x})\Phi(f_{\x})=\langle v_{\x},f_{\x}\rangle\,\Phi(f_{\x})$ and the denseness of the linear span of $\{\Phi(f_{\x}),\ f_{\x}\in D_{\x}\}$ (see \cite[Proposition 6.2]{LSTT}).
\par 
Hence  Theorem \ref{sb} applies and defines a self-adjoint operator $\widehat H_{T}$ for any symmetric operator $T$ which is $\widehat H_{0}$-bounded with relative bound ${\widehat a}<1$. By Remark \ref{suit}, $T$ should be a suitable regularization of the ill-defined operator $-AH^{-1}_{\rm free}A^{*}$; for $A$ given in \eqref{ann}, the right choice, consisting in a regularization of the diagonal (with respect to the direct sum structure of $\H$ in \eqref{Fock}) part of $-AH^{-1}_{\rm free}A^{*}$, is provided in \cite[equations  (29)-(32)]{LS}. Here we denote such an operator by $T=T_{\rm Nelson}$; it is infinitesimally $\widehat H_{0}$-bounded by \cite[Lemma 3.10]{LS} (let us notice that, by \eqref{Rl}, our $\widehat H_{0}$ coincides with the operator there defined as $(1-G^{*})L(1-G)$).
\par
Given the sequence $v_{n}\in L^{2}(\RE^{3})$ with Fourier transform $\hat v_{n}=\chi_{n}\hat v$, 
where $\chi_{n}$ denotes the characteristic function of a ball of radius $R=n$ (this provides an ultraviolett cutoff on the boson momenta), let us denote by $A_{n}$ the sequence of operators in $L^{2}(\RE^{3N})\otimes\H_{b}$ defined as $A$ in  \eqref{ann} with $v$ replaced by $v_{n}$. One has that $A_{n}$ is closed, $A_{n}\in\B(\dom(H^{1/2}_{\rm free}),L^{2}(\RE^{3N})\otimes\H_{b})$ and $A_{n}^{*}+A_{n}$ is infinitesimally $H_{\rm free}$-bounded (see, e.g., \cite[Section 14.5.1]{Arai}, \cite[Appendix B]{GW2}) and so such $A_{n}$'s fit to the hypotheses in Theorem \ref{bb}. By \cite[Proposition 3.2]{Sch1} (see also the proof  of \cite[Theorem 1.4]{LS}), one has $\|(A_{n}H^{-1}_{\rm free})^{*}-(AH^{-1}_{\rm free})^{*}\|_{\H,\H}\to 0$, which is equivalent to \eqref{Aenne}. Let 
$E_{n}$ be the sequence of bounded symmetric operators in $L^{2}(\RE^{3N})\otimes\H_{b}$ corresponding to the multiplication by the real constant given by (minus) the leading order term in the expansion in the coupling constant $g$ of the the ground state energy at zero total momentum of the regularized Hamiltonian $H_{\rm free}+A^{*}_{n}+A_{n}$ (see, e.g., \cite[Section 19.2]{Spohn}):
$$
E_{n}:=g^{2}N\,\big\|\big(-\Delta_{(3)}+(-\Delta_{(3)}+1)^{1/2}\big)^{-1/2}v_{n}\big\|^{2}_{L^{2}(\RE^{3})}=g^{2}N\int_{\RE^{3}}\frac{|\hat v_{n}(\kappa)|^{2}}{|\kappa|^{2}+(|\kappa|^{2}+1)^{1/2}}\ d\kappa\,.
$$
Defining then
$$
T_{n}:=E_{n}-A_{n}H^{-1}_{\rm free}A^{*}_{n}\,,
$$
by \cite[Proposition 3.1]{Sch1} (see also the proof of Theorem 1.4 in \cite{LS}), one has $T_{n}\to T_{\rm Nelson}$ in norm as operators in $\B(\D(T_{\rm Nelson}),L^{2}(\RE^{3N})\otimes \H_{b})$; thus hypothesis \eqref{Hy2} holds. Hypothesis \eqref{Hy1} holds since the estimates in \cite{LS} with $\hat v$ replaced by $\hat v_{n}$ are bounded by the integrals with $\hat v$ (see in particular the arguments given in the proof of  \cite[Theorem 1.4]{LS}). Therefore, by Theorem \ref{bb},  
$$
\lim_{n\uparrow \infty}(H_{\rm free}+A^{*}_{n}+A_{n}+E_{n})=H_{\rm Nelson}:=
\overline{H}_{\rm free}+A^{*}+A_{T_{\rm Nelson}}
\quad\text{in norm resolvent sense}
$$ 
and so the self-adjoint Hamiltonian $H_{\rm Nelson}$ provided by Theorem \ref{sb} with $T=T_{\rm Nelson}$ coincides with the one given by Nelson in \cite{N} (this is our version of \cite[Theorem 1.4]{LS}; see also \cite[Proposition 2.4]{Sch1}).\par By Theorem \ref{sb},
$$
\dom(H_{\rm Nelson})=\{\Psi\in \dom(H_{\rm free}^{1-s}):\Psi+(AH_{\rm free}^{-1})^{*}\Psi\in\dom(H_{\rm free})\}\,,\qquad s>\frac34\,,
$$
and
\be\label{nelson}
(-H_{\rm Nelson}+z)^{-1}=(-H_{\rm free}+z)^{-1}-
\begin{bmatrix}G_{z}&R_{z}\end{bmatrix}\begin{bmatrix}A_{{T_{\rm Nelson}}
} G_{z}&G^{*}_{{\bar z} }-\uno\\G_{z}-\uno&R_{z}
\end{bmatrix}^{\!-1}\begin{bmatrix}G_{{\bar z} }^{*}\\R_{z}\end{bmatrix}\,,
\ee
where  $R_{z}:=(-H_{\rm free}+z)^{-1}$, $G_{z}:=(AR_{{\bar z} })^{*}$ and 
$A_{{T_{\rm Nelson}}}G_{z}=T_{\rm Nelson}-A(G-G_{z})$. Notice that, since the operator sequence $A_{n}\big(H_{\rm free}^{-1}A_{n}^{*}+(-H_{\rm free}+{z})^{-1}A_{n}^{*}\big)$ converges to $-A(G-G_{z})$ in $\B(L^{2}(\RE^{3N})\otimes\H_{b})$, one has that $A_{{T_{\rm Nelson}}}G_{z}$ is the limit of $E_{n}+A_{n}(-H_{\rm free}+{z})^{-1}A_{n}^{*}$ as operators in $\B(\D(T_{\rm Nelson}),L^{2}(\RE^{3N})\otimes \H_{b})$.
\end{subsection} 
\end{section}

\end{document}